\documentclass[11pt]{amsart}
\usepackage[utf8]{inputenc}
\usepackage[english]{babel}
\usepackage{amsmath,amssymb,amsfonts,amsthm, mathrsfs}
\usepackage{amscd} 
\usepackage{mathtools} 
\usepackage{longtable}
\usepackage[noadjust]{cite} 
\usepackage{enumitem}
\usepackage{float}
\usepackage[mathcal]{euscript} 
\usepackage[unicode]{hyperref}
\usepackage{ragged2e}
\usepackage{xfrac} 

\usepackage{graphicx, array}
\graphicspath{{figures/}}
\usepackage[labelfont=small,labelformat=simple]{subcaption}

\usepackage{xcolor}

\theoremstyle{plain}
\newtheorem{theorem}{Theorem}[section]

\theoremstyle{definition}
\newtheorem{definition}[theorem]{Definition}

\newtheorem{remark}[theorem]{Remark}

\makeatletter
\def\@setthanks{\vspace{-\baselineskip}\def\thanks##1{\@par##1\@addpunct.}\thankses}
\makeatother

\begin{document}

\renewcommand{\abstractname}{Abstract}
\renewcommand{\refname}{References}
\renewcommand{\tablename}{Table}
\renewcommand{\figurename}{Figure}
\renewcommand{\proofname}{Proof}

\title[Euler's two-center problem]{Scattering invariants in Euler's two-center problem}

\author[N. Martynchuk]{N. Martynchuk$^1$}

\author[H. R. Dullin]{H. R. Dullin$^2$}

\author[K. Efstathiou]{K. Efstathiou$^1$}

\author[H. Waalkens]{H. Waalkens$^1$}

\thanks{ $^1$ Johann Bernoulli Institute for Mathematics and Computer Science,
  University of Groningen, P.O. Box 407, 9700 AK Groningen, The Netherlands}
\thanks{ $^2$ School of Mathematics and Statistics, The University of Sydney, Sydney, NSW 2006, Australia}


  
  \begin{abstract}
  The problem of two fixed centers was introduced by Euler  as early as in 1760. It plays an important role both in celestial mechanics and in the microscopic world. In the present paper we study the spatial 
  problem in the case of arbitrary  (both positive and negative)  strengths of the centers. Combining techniques from scattering theory and Liouville integrability, 
  we show that this spatial problem has
  topologically non-trivial scattering dynamics,  which we identify as scattering monodromy. 
  The approach that we introduce in this paper applies more generally
  to scattering systems that are integrable in the Liouville sense.\\
  \vspace{-2mm}
\end{abstract}



\maketitle


\section{Introduction}

The problem of two fixed centers, also known as the Euler $3$-body problem, is one of the most fundamental integrable problems of classical mechanics. It 
describes the motion of a point particle in Euclidean space under the influence of the Newtonian force field 
$$F = - \nabla V, \ V = - \dfrac{\mu_1}{r_1} - \dfrac{\mu_2}{r_2}.$$
Here $r_i$ are the distances of the particle to the two fixed centers and $\mu_i$ are the strengths (the masses or the charges) of these centers. 
We note that the Kepler problem corresponds to the special cases when
the centers coincide or when one of the strengths is zero.

The (gravitational) Euler problem  was first studied by L. Euler in a series of works in the 1760s \cite{Euler1760, Euler1766, Euler1767}. He discovered that this problem is
integrable by putting the equations of motion in a separated form. Elliptic coordinates, which separate the problem and which are now commonly used, appeared in
his later paper \cite{Euler1767} and, at about the same time, in the work of Lagrange \cite{Lagrange1766-69}. 
The systematic use of
elliptic coordinates in classical mechanics was initiated by Jacobi, who used a more general form of these coordinates to integrate, among other systems, the geodesic flow on 
a triaxial ellipsoid; see \cite{Jacobi1884} for more details.

Since the early works of Euler and Lagrange  the Euler problem and its  generalizations have been studied by many authors. First classically and then,
since the works of Pauli \cite{Pauli1922} and Niessen \cite{Niessen1923} in the early 1920s, 
also in the setting of quantum mechanics.
We indicatively mention the works \cite{Charlier1902, Whittaker1917, Erikson1949, Deprit1962, Vosmischera2003, WaalkensDullinRichter04, Dullin2016, Seri2016}.
For a historical overview we refer to \cite{OMathuna2008, Gerasimov2007}.

In the present work we will be interested in the spatial Euler problem. 
For us, it will be important that this problem is a Hamiltonian system with two additional structures: it is a
\textit{scattering system} and
it is also \textit{integrable in the Liouville sense}. The structure of a scattering system comes from the fact that the potential  $$V(q) \to 0, \ \|q\| \to \infty,$$ 
decays at infinity sufficiently fast (is of \textit{long range}). 
It allows one to compare a given set of initial conditions  at $t = -\infty$ with the outcomes at $t = +\infty$. An introduction to the general theory of scattering systems can be found in \cite{Derezinski2013, Knauf2011}.
Liouville integrability comes from 
the fact that the system is \textit{separable};
the three commuting integrals of motion are:
\begin{itemize}
 \item the energy function --- the Hamiltonian,
 \item the separation constant; see Subsection~\ref{sec/separation_and_integrability},
 \item the component of the angular momentum about the axis connecting the two centers.
\end{itemize}
An introduction to the general theory of  Liouville integrable systems can be found in 
in \cite{Bolsinov2004, Cushman2015, Knauf2011}.

Separately these two structures of the Euler problem have been  discussed in the literature. Scattering has been studied, for instance,
in \cite{Klein2008, Seri2016}.  The corresponding 
Liouville fibration has been studied in \cite{WaalkensDullinRichter04} --- from the perspective 
of Fomenko theory \cite{Fomenko1990, Bolsinov2004}, action coordinates and \textit{Hamiltonian monodromy} \cite{Duistermaat1980}.
We will consider both of the structures together and show that the Euler  problem has non-trivial scattering invariants, which we will call \textit{purely scattering} and \textit{mixed}
\textit{scattering monodromy}, cf.
\cite{Knauf1999, Bates2007, DullinWaalkens2008, Martynchuk2016, Efstathiou2017}.
For completeness,
the qualitatively different case of Hamiltonian monodromy 
will be also discussed. 
We note that the approach that we introduce in the present paper applies more generally
  to systems that are both scattering and integrable in the Liouville sense.

The paper is organized as follows. The problem is introduced in Section~\ref{sec/preliminaries}. Bifurcation diagrams  are given in 
Section~\ref{sec/bifurcationdiagrams}. In  Section~\ref{sec/scattering} we discuss 
classical potential scattering theory. In Section~\ref{sec/scattering_in_integrable_systems} we adapt the discussion of Section~\ref{sec/scattering} to the context of scattering systems that are
integrable in the 
Liouville sense. In particular, we give a definition of a reference system for integrable systems. 
We note that
the choice of a reference system is important for the definition of scattering monodromy; see Subsection~\ref{subsection/scattering_invariants}. 
For the Euler problem, scattering monodromy is discussed in detail in Section~\ref{sec/scattering_2cp}.
Hamiltonian monodromy  is addressed in Subsection~\ref{sec/topology}. The main part of the paper is concluded with a discussion in Section~\ref{sec/discussion}. 
Additional details are presented in the Appendix.


\section{Preliminaries} \label{sec/preliminaries}

We start with the $3$-dimensional Euclidean space $\mathbb R^3$ and two distinct  points in this space, denoted by $o_1$ and $o_2.$ Let $q = (x,y,z)$ be Cartesian coordinates in $\mathbb R^3$ 
and let $p = (p_x,p_y,p_z)$ be the conjugate momenta in $T^{*}_q\mathbb R^3$.
The \textit{Euler two-center problem}
can be defined 
as a 
Hamiltonian system on $T^{*}(\mathbb R^3 \setminus \{o_1, o_2\})$ with a Hamiltonian function $H$ given by
\begin{equation} \label{eq/hamiltonian}
 H = \frac{\|p\|^2}{2} + V(q), \ \ V(q) = - \dfrac{\mu_1}{r_1} - \dfrac{\mu_2}{r_2},
\end{equation}
where $r_i \colon \mathbb R^3 \to \mathbb R$ is the distance to the center $o_i$. The strengths of the centers $\mu_i$ can be both positive and negative; without loss of generality
we assume that the center $o_1$ is stronger, that is,
$|\mu_2|\le |\mu_1|$.

\begin{remark}
When $\mu_i > 0$ (resp., $\mu_i < 0$) the center $o_i$ is attractive (resp., repulsive). 
The cases $\mu_1 \ne \mu_2 = 0$ and $\mu_2 \ne \mu_1 = 0$
correspond to a Kepler problem. In the 
case $\mu_1 = \mu_2 = 0$ the dynamics is trivial and we have the free motion $(t, q_0,p_0) \mapsto (q_0 + tp_0, p_0)$. 
\end{remark}

\subsection{Separation and integrability} \label{sec/separation_and_integrability}

 Without loss of generality we assume
$o_i = (0,0,(-1)^{i}a)$ for some $a > 0$, so that, in particular, the fixed centers $o_{1}$ and $o_2$ are located on the $z$-axis in the configuration space.
Rotations around the $z$-axis leave the potential function $V$ invariant. It follows that (the $z$-component of) the angular momentum
\begin{equation} \label{eq/momentum}
L_z = xp_y - yp_x
\end{equation}
commutes with $H$, that is, $L_z$ is a first integral. It is known \cite{Whittaker1917, Erikson1949} that there exists another 
first integral given by 
\begin{equation} \label{eq/thirdintegral}
G = H + \frac{1}{2}(L^2-a^2(p_x^2+p_y^2)) + a(z+a)\frac{\mu_1}{r_1}- a(z-a)\frac{\mu_2}{r_2}, 
\end{equation}
where $L^2 = L_x^2 + L_y^2 + L_z^2$ is the squared angular momentum.
The expression for the integral $G$ can be obtained using separation in elliptic coordinates, as described below.
It will follow from the separation procedure that the function $G$  commutes both with $H$ and with $L_z$,  which means that the problem of two fixed centers 
is \textit{Liouville integrable}.

Consider prolate ellipsoidal coordinates 
$(\xi, \eta, \varphi)$:
\begin{equation} \label{equation/ellcoord}
\xi = \frac{1}{2a}(r_1+r_2), \ \ 
\eta = \frac{1}{2a}(r_1-r_2), \ \ 
\varphi = \textup{Arg}(x + iy).  
\end{equation}
Here $\xi \in [1, \infty), \, \eta \in [-1,1]$, and $\varphi \in \mathbb R / 2\pi \mathbb Z.$ Let $p_{\xi}, p_\eta, p_\varphi = L_z$ 
be the conjugate momenta and $l$ be the value of $L_z$.
In the new coordinates the Hamiltonian $H$ has the form
\begin{equation} \label{eq/hamnewcoord}
H = \dfrac{H_\xi + H_\eta}{\xi^2 - \eta^2},
\end{equation}
where
\begin{equation*}
 H_\xi = \frac{1}{2a^2}(\xi^2-1)p_\xi^2 + \frac{1}{2a^2}\frac{l^2}{\xi^2-1} - \frac{\mu_1+\mu_2}{a}\xi
\end{equation*}
and
\begin{equation*}
  H_\eta = \frac{1}{2a^2}(1-\eta^2)p_\eta^2 + \frac{1}{2a^2}\frac{l^2}{1-\eta^2} + \frac{\mu_1-\mu_2}{a}\eta.
\end{equation*}
Multiplying Eq.~\eqref{eq/hamnewcoord} by $\xi^2-\eta^2$ and separating we get the first integral
\begin{equation*}
G =  \xi^2H - H_{\xi} = \eta^2 H + H_\eta.
\end{equation*}
In original coordinates $G$ has the form given in Eq.~\eqref{eq/thirdintegral}.
Since $L_z = p_\varphi$, the function $G$ commutes both with $H$ and with $L_z$.

\subsection{Regularization} \label{sec/regularization}
We note that in the case when one of the strengths is attractive, collision orbits are present and, consequently, the flow of $H$ on $T^{*}(\mathbb R^3 \setminus \{o_1, o_2\})$ is not complete. This complication is, however,
not essential for our analysis since collision orbits, as in the Kepler case, can be regularized.
More specifically, there exists a $6$-dimensional symplectic manifold $(P, \omega)$ and a smooth Hamiltonian function $\tilde{H}$ on $P$ such that
\begin{enumerate}
 \item $(T^{*}(\mathbb R^3 \setminus \{o_1, o_2\}, dq \wedge dp)$ is symplectically embedded in $(P, \omega)$,
 \item $H = \tilde{H}|_{T^{*}(\mathbb R^3 \setminus \{o_1, o_2\})}$,
 \item The flow of $\tilde{H}$ on $P$ is complete.
\end{enumerate}
This result is essentially due to \cite[Proposition 2.3]{Klein2008}, where a similar statement is proved for the gravitational planar problem. The planar problem in the case of arbitrary strengths 
can be treated similarly (note that collisions with 
a repulsive center are not possible). The spatial
case follows from the planar case
since 
collisions  occur only when $L_z = 0$. 
We note that the integrals $L_z$ and $G$ can be also extended to $P$. 

One important property of the regularization is that the extensions of the integrals to $P$, which will be also denoted by $H$, $L_z$ and $G$, form a \textit{completely integrable system}. In particular,
the Arnol'd-Liouville theorem \cite{Arnold1968} applies. In what follows
we shall work on the regularized space $P$.


\section{Bifurcation diagrams} \label{sec/bifurcationdiagrams}
Before we move further and discuss scattering in the Euler problem, we shall compute the \textit{bifurcation diagrams} of the integral map $F = (H,L_z,G)$, that is, the set 
of the critical values of this map. 
We distinguish two cases, depending on whether $L_z$ is zero or different from zero. The bifurcation diagrams are  obtained by superimposing the critical values found in these two cases. By a choice of units 
we assume that 
$a = 1$.
\subsection{The case $L_z = 0$} \label{subsection/the_planar_case}

Since $L_z = 0$, the motion is planar. We assume that it takes place in the $xz$-plane. Consider the  elliptic coordinates $(\lambda, \nu) \in \mathbb R \times S^1[-\pi,\pi]$ defined by 
$$x = \sinh \lambda \cos \nu, \ z = \cosh \lambda \sin \nu.$$ 
The level set of constant $H = h, L_z = l = 0$ and $G = g$ in these coordinates is given by the equations
\begin{align*}
p^2_\lambda &= 2 h \cosh^2 \lambda + 2 (\mu_1+\mu_2) \cosh \lambda - 2g, \\
p^2_\nu &= -2 h \sin^2 \nu - 2 (\mu_1-\mu_2) \sin \nu + 2g,
\end{align*}
where $p_\lambda$ and $p_\nu$ are the momenta conjugate to $\lambda$ and $\nu.$
The value $(h,0,g)$ is critical when the Jacobian matrix corresponding to these equations does not have a full rank. Computation yields  lines 
\begin{multline*}
\ell_1 = \{g = h+\mu_2-\mu_1, l = 0\}, \  \ell_2 = \{g = h+\mu_1-\mu_2, l = 0\} \\ \mbox{ and } \  \ell_3 = \{g = h+\mu, l = 0\}, \ \mu = \mu_1+\mu_2,
\end{multline*}
and two curves
\begin{align*}\{g &= \mu \cosh\lambda/2, \ h =  -\mu /2\cosh\lambda, \ l = 0\}, \\
\{g &= (\mu_1 - \mu_2) \sin \nu/2, \ h = (\mu_2 - \mu_1)/2 \sin \nu, \ l = 0\}.
\end{align*}
Points that do not correspond to any physical motion must be removed from the obtained set (allowed motion corresponds to the regions where the squared momenta are positive).

\begin{remark}
The corresponding diagrams in the planar problem are given in Appendix~\ref{appendix/bifurcation_diagrams_for_the_planar_problem}; see Fig.~\ref{bifd_planar_generic_case} and \ref{bifd_planar_critical_case}. We note that in the planar case the set of the regular values of $F$ consists of contractible
components and hence the topology of the regular part of the Liouville fibration is trivial. Interestingly, this is not the case if the dimension of the configuration space is $n = 3$. 

We note that the singular Liouville foliation has non-trivial topology already in the planar case.  The corresponding bifurcations, in the sense of Fomenko theory 
\cite{Fomenko1986, Fomenko1987, Fomenko1990, BMF1990, Bolsinov2004}, have been studied in \cite{WaalkensDullinRichter04, Kim2017}.
\end{remark}


\subsection{The case $L_z \ne 0$}

\begin{figure}[ht]
\begin{center}
\includegraphics[width=1\linewidth]{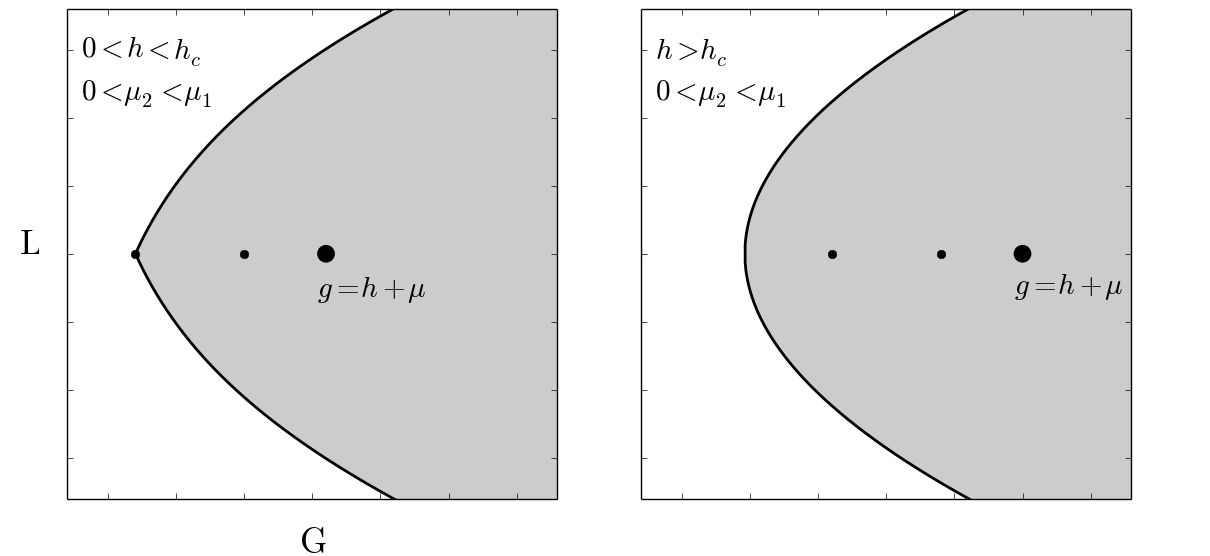}
\end{center}
\caption{Positive energy slices of the bifurcation diagram for the spatial Euler problem, attractive case. The black points correspond to the critical lines $\ell_i$.} 
\label{bifd_spatial_attr}
\end{figure}

In order to compute the critical values in this case
it is convenient to use the ellipsoidal coordinates $(\xi, \eta)$. (We note that for $L_z \ne 0$ the $z$-axis is inaccessible, so $(\xi, \eta)$ are non-singular.)
The level set of constant $H = h, L_z = l$ and $G = g$ in these coordinates is given by the equations
\begin{align*}
p^2_\xi &= \dfrac{(\xi^2 - 1)(2h\xi^2 + 2(\mu_1+\mu_2)\xi - 2g) - l^2}{(\xi^2-1)^2}, \\
p^2_\eta &= \dfrac{(1-\eta^2)(-2h\eta^2 - 2(\mu_1-\mu_2)\eta + 2g) - l^2}{(1-\eta^2)^2}.
\end{align*}
The value $(h,l,g)$ with $l \ne 0$ is critical when the corresponding Jacobian matrix does not have a full rank. 
Computation yields the following sets of critical values:
$$\left\{g = h (2\xi^2 - 1) + \dfrac{(\mu_1 + \mu_2) (3 \xi^2 - 1)}{2 \xi}, l^2 = 
-\dfrac{(\mu_1 + \mu_2 + 2 h \xi) (-1 + \xi^2)^2}{\xi}\right\}, $$
$$
\left\{g = h (2 \eta^2 - 1) + \dfrac{(\mu_1 - \mu_2) (3 \eta^2 - 1)}{2 \eta} , l^2 = -\dfrac{(\mu_1 - \mu_2 + 2 h \eta) (-1 + \eta^2)^2}{\eta}\right\},
$$
where $\xi > 1$ and $-1 < \eta < 1$. As above,  points that do not correspond to any physical motion must be removed. 

Representative positive energy slices in the gravitational case $0 < \mu_2 < \mu_1$  are given in Fig.~\ref{bifd_spatial_attr}. The case of arbitrary strengths $\mu_i$ is similar. 
The structure of the corresponding  diagrams can partially be deduced from the diagrams obtained in the planar case; 
see Appendix~\ref{appendix/bifurcation_diagrams_for_the_planar_problem}. 



\section{Classical scattering theory} \label{sec/scattering}

In this section we discuss certain qualitative aspects of scattering theory following \cite{Knauf1999, Knauf2011}. In Section~\ref{sec/scattering_in_integrable_systems} we explain how the theory
can be adapted to the context of scattering systems that are integrable in the Liouville sense, with the Euler problem as the leading example. 

\subsection{Preliminary remarks}

Classical scattering theory goes back to the works of
Cook \cite{Cook1967}, Hunziker \cite{Hunziker1968} and  Simon \cite{Simon1971}.  Since then it
has received considerable interest and has been actively developed in several directions; see \cite{Herbst1974, Knauf1999, Derezinski2013, Bates2007, DullinWaalkens2008}.
 
In the framework of classical scattering one considers  two Hamiltonian functions $H$ and $H_r$ such that their flows become similar `at infinity'. 
This allows one can compare a given distribution of particles, that is, initial conditions,
at $t = -\infty$
 with 
their final distribution at $t = +\infty$. 
To be more specific, consider a pair of  Hamiltonians  on $T^{*}\mathbb R^n$ given by
$$H = \frac{1}{2} \|p\|^2  + V(q) \mbox{ \ and \ } H_r = \frac{1}{2} \|p\|^2+V_r(q),$$
where the (singular) potentials $V$ and $V_r$ are  assumed to
satisfy a certain 
decay assumption; see Subsection~\ref{subsection/decay_assumptions}. 
For scattering Hamiltonians the comparison will be achieved in two steps.
First we shall parametrize the possible initial and final distributions using the flow of the `free' Hamiltonian $H_0 = \frac{1}{2} \|p\|^2$. Then, for a given invariant manifold, we shall construct the \textit{scattering map},
where only $H$ and $H_r$ are compared. 

\begin{remark}
One reason for such a procedure is the following.  As we shall see later in Section~\ref{sec/scattering_in_integrable_systems} and Appendix~\ref{appendix/proof_reference_systems}, 
the `free' Hamiltonian is not a natural reference Hamiltonian for the 
Euler problem, unless the strengths  $\mu_1 = \mu_2$. However, the `free' Hamiltonian  will be convenient for the definition of the asymptotic states. 
\end{remark}

\begin{remark}
In what follows we sometimes refer to $H, H_r$ as \textit{scattering Hamiltonians} and to $H_r$ is a \textit{reference Hamiltonian} for $H$. We note that the `reference' dynamics of $H_r$ is usually
chosen to be simpler than the `original' dynamics of $H$.
\end{remark}

\subsection{Decay assumptions} \label{subsection/decay_assumptions}
In classical potential scattering the potential function $V \colon \mathbb R^n \to \mathbb R$ of a Hamiltonian $H = \frac{1}{2}\|p\|^2 + V(q)$ is assumed to decay according to one of the following estimates:
 \begin{itemize}
   \item[1.] Finite-range: $\textup{supp}(V) \subset \mathbb R^n$ is compact;

   \item[2.] Short-range case:  $|\partial_k V(q)| < c {(\|q \|+  1)}^{-|k|-\varepsilon};$

   \item[3.] Long-range case:  $|\partial_k V(q)\| < c {(\| q \|+1)}^{-|k|-1-\varepsilon}.$
  \end{itemize}
Here  $c $ and $ \varepsilon$ are positive constants, $k = (k_1, \ldots, k_n) \in \mathbb N_0^n$ is a multi-index, $|k| = k_1 + \ldots +k_n$ is a norm of $k$ and
$\|q\|$ denotes the Euclidean norm of $q$.
For 
instance, any Kepler potential is of long range and the same is true of the potential found in
the Euler problem. 

We will assume the potentials
$V$ and $V_r$ are \textit{short-range} relative to some decaying rotationally symmetric potentials $\widetilde{V}$ and $\widetilde{V}_r$, respectively. For the potential $V$ this means that 
$$
\|\partial_k V - \partial_k \widetilde{V} \| < c {\|q \|}^{-|k|-1-\varepsilon},
$$
where  $\widetilde{V}$ is rotationally symmetric 
with $\widetilde{V}(q) \to 0, \|q\| \to \infty.$  A similar estimate is assumed to hold for $V_r - \widetilde{V}_r.$

\begin{remark}
The potentials $\widetilde{V}$ and $\widetilde{V}_r$    are needed 
to guarantee that the asymptotic direction and the  
 impact parameter are defined and parametrize the scattering trajectories in a continuous way. This is known to be the case for short-range potentials $V$ \cite{Knauf2011}. 
 Our case reduces to the case of symmetric potentials and in that case the statement follows from the conservation of the angular momentum.
 \end{remark}

\subsection{Asymptotic states} \label{subsection/asymptotic_states}
The Hamiltonian flow $g_H^t \colon P \to P$ of $H$ partitions the (regularized) phase space $P$ into the following invariant subsets:
$$b^{\pm} = \{x \in P \mid \textup{sup}_{t \in \mathbb R^{\pm}} \|g_H^t(x)\| < \infty\} \ \mbox{ and } \ s^{\pm} =  
\{x \in P \mid H(x) > 0\} \setminus b^{\pm}.$$
The invariant subsets
$$b = b^{+} \cap b^{-}, \ s = s^{+} \cap s^{-} \  \mbox{ and } \ trp = (b^{+} \setminus b^{-}) \cup (b^{-} \setminus b^{+})$$ are the sets of the \textit{bound},
the \textit{scattering} and the \textit{trapped} states, respectively.  We note that $s^{-}, s^{+}$ and hence $s = s^{-} \cap s^{+}$ are open subsets
of $P$. 


If the potential  $V$ is short-range relative to a decaying rotationally symmetric potential, then 
the following limits
$$
 {\hat{p}}^{\pm}(x) = \lim_{t\to \pm \infty}  p(t,x)  \ \mbox{ and } \ 
 q_{\bot}^{\pm}(x) = \lim_{t\to \pm \infty} \left( q(t,x) - \langle q(t,x),  {\hat{p}}^{\pm}(x)\rangle \frac{{\hat{p}}^{\pm}(x)}{2h} \right),
$$
where $h = H(x) > 0$ is the energy of $g_H^t(x)$, are defined for any $x \in s^{\pm}$ and depend continuously on $x$. These limits are
called the
\textit{asymptotic direction}
and the \textit{impact parameter} of the trajectory $g_H^t(x)$, respectively. We note that 
an asymptotic direction is always orthogonal to the corresponding impact parameter. Due to the  $g_H^t$-invariance of ${\hat{p}}^{\pm}$ and $q_{\bot}^{\pm}$, 
we have the maps
$$
A^{\pm} = (\hat{p}^{\pm},q_{\bot}^{\pm}) \colon  s / g^t_H \to  AS 
$$
from $s / g^t_H$ to the asymptotic states $AS \subset \mathbb R^{n}\times\mathbb R^n$. Here $s / g^t_H$ is the space of trajectories in $s$, that is, it is a quotient space of $s$ by
the equivalence relation where
two points are considered equivalent if and only if they belong to a single trajectory $g^t_H(x)$. 
Similarly, one can construct
the maps 
$$
A_r^{\pm} = (\hat{p}^{\pm},q_{\bot}^{\pm}) \colon  s_r / g^t_{H_r} \to AS
$$
for the `reference' Hamiltonian $H_r =  \frac{1}{2}p^2 + V_r(q).$ 

\subsection{Scattering map} \label{subsection/scattering_map}
Using the maps $A^{\pm}$ and $A^{\pm}_r$ constructed in Subsection~\ref{subsection/asymptotic_states}, we can now define the notion of a \textit{scattering map} for a given invariant
submanifold $R$ of $s$.

\begin{definition} \label{definition/scattering_map} Let $R$ be a $g_H^t$-invariant submanifold of $s$ and $B = R / g_H^t$. Assume that the composition map
 $$
 S = (A^{-})^{-1} \circ A^{-}_r \circ (A^{+}_r)^{-1} \circ A^{+}
 $$
is well defined and maps $B$ to itself. The map $S$ is called the \textit{scattering map} (w.r.t. $H, H_r$ and $B$).
\end{definition}

\begin{remark}
Due to the decay assumptions the maps 
 $$A^{\pm} \colon s/g_H^t \to AS \ \mbox{ and } \  A_r^{\pm} \colon s_r/g_H^t \to AS$$ are homeomorphisms onto their images in $AS$. 
 It follows that the scattering map $S \colon B \to B$ is a homeomorphism as well. Here the sets
 $s/g_H^t$, $s/g_H^t$ and $B$ are endowed with the quotient topology.

\end{remark}

\subsection{Knauf's topological degree}
To get qualitative information about the scattering it is useful to look at topological invariants of the scattering map.
An important example in the context of general scattering theory is \textit{Knauf's topological index}; see \cite{Knauf1999, Knauf2008}. We shall now recall its definition. 

Consider the case when the potential $V$ is short-range relative to $V_r = 0$. 
Let $h > 0$ be 
a \textit{non-trapping energy}, that is, a positive energy such that the energy level $H^{-1}(h)$ contains no trapping states, and let $R = H^{-1}(h) \cap s$ be the intersection of the level
$H^{-1}(h)$ with the set $s$ of the scattering states. There is the following result.

\begin{theorem} \textup{(\cite{Derezinski2013, Knauf1999})} 
The scattering manifold $B = R / g^t_H$ is the cotangent bundle $T^*S^{n-1}$, where $S^{n-1}$ is the sphere of asymptotic directions. The corresponding scattering map 
$$S_h \colon B \to B$$
is a symplectic transformation of $T^*S^{n-1}$.
\end{theorem}

Knauf's topological degree is defined as a topological invariant of $S_h$. Specifically,
let  $\Pr \colon T^{*}S^{n-1} \to S^{n-1}$ be the canonical projection  and
$$S^{n-1}_{p} = T^{*}_pS^{n-1} \cup \{ \infty \}$$ 
be the one-point compactification of the cotangent space $T^{*}_pS^{n-1}$. 

\begin{definition} \textup{(Knauf, \cite{Knauf1999})}
 The \textit{degree  $\textup{deg}(h)$ of the energy $h$ scattering map} $S_h$ is defined as the topological degree of the map
\begin{equation*}
 \Pr \circ S_h \colon  S^{n-1}_{p} \to S^{n-1}.
\end{equation*}
\end{definition}
\begin{remark}
We note that by  continuity $\textup{deg}(h)$ is
independent of the choice of the initial direction $p \in S^{n-1}$; see \cite{Knauf1999}.
\end{remark}

The following theorem shows that for regular (that is, everywhere smooth) potentials $\textup{deg}(h)$ is either $0$ or $1$,
depending on the value of the energy $h$; see Fig.~\ref{degree}. We note that for singular potentials, such as the Kepler potential, values different from $0$ and $1$ may appear. 
\begin{theorem} \label{kktheorem} \textup{(Knauf-Krapf, \cite{Knauf2008})} 
 Let $V$ be a regular short-range potential and $h > 0$ be a non-trapping energy. Then 
 $$
 \textup{deg}(h) = \begin{cases}
		    0,  \ \ h \in (\sup V, \infty), \\
                    1, \ \ h \in (0, \sup V).
                    
                   \end{cases}
$$
\begin{figure}[ht]
\begin{center}
\begin{minipage}[h]{0.49\linewidth}
\includegraphics[width=1\linewidth]{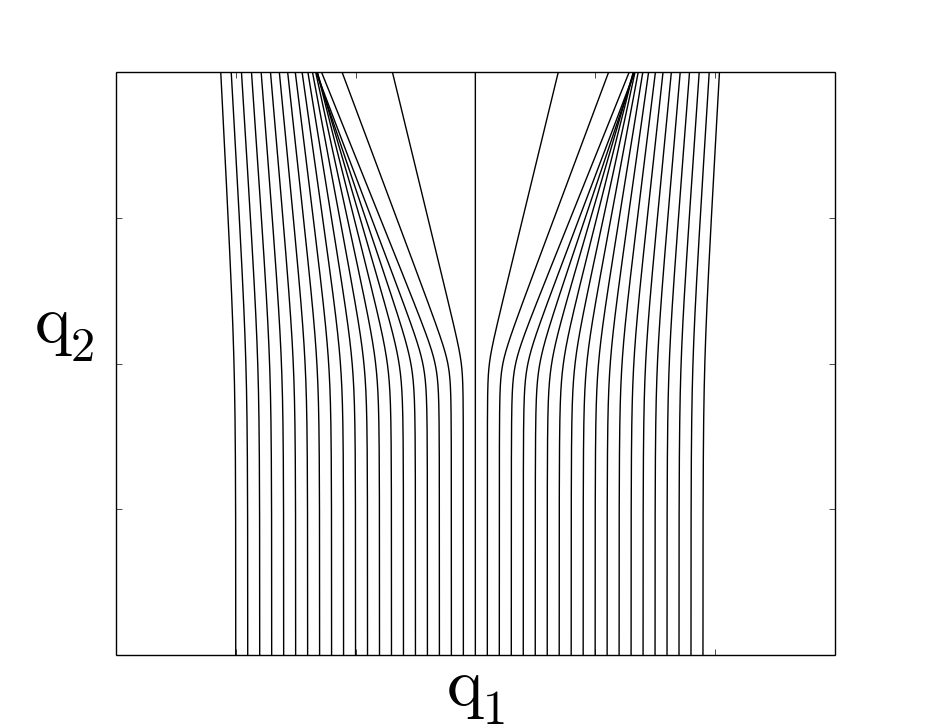}
\end{minipage}
\hfill
\begin{minipage}[h]{0.48\linewidth}
\includegraphics[width=1\linewidth]{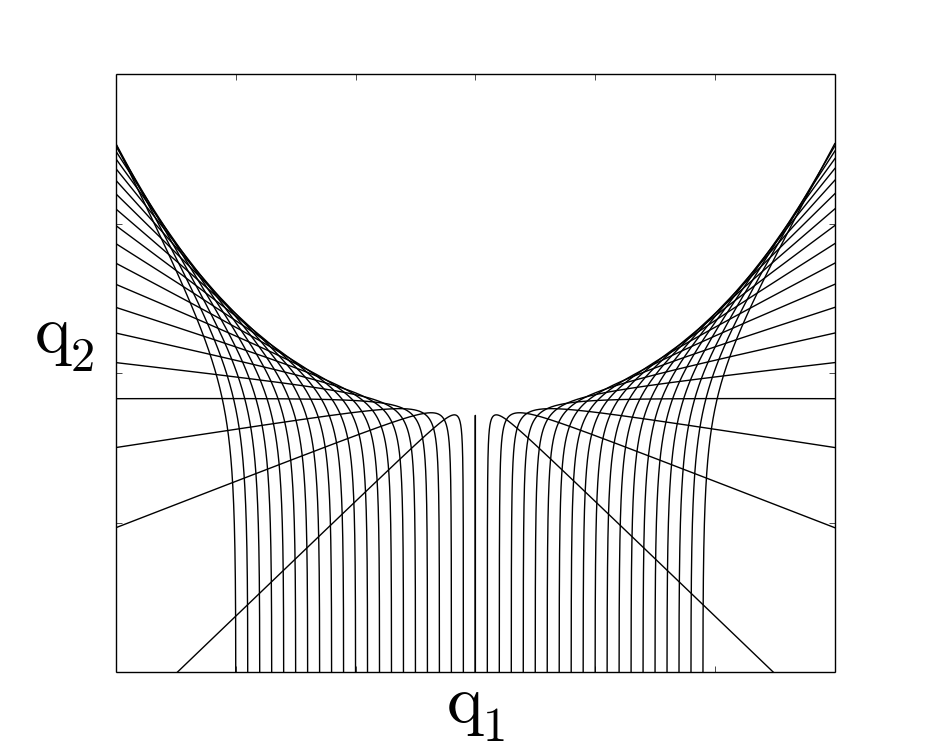}
\end{minipage}
\end{center}
\caption{Scattering at different energies.  At high energies
 $\textup{deg}(h) = 0$ (left), at low energies  $\textup{deg}(h) = 1$ (right).  } 
\label{degree}
\end{figure} 
\end{theorem}

\begin{remark}
 For the Euler problem with $\mu_1\mu_2 \ne 0$, Knauf's degree is not defined (every positive energy $h$ is trapping).
 Moreover, the free flow is not a proper reference 
 unless $\mu_1 = \mu_2$; see Section~\ref{sec/scattering_in_integrable_systems}. Nonetheless,  as we shall show in Sections~\ref{sec/scattering_in_integrable_systems} and \ref{sec/scattering_2cp}, for a proper choice of a reference 
 Hamiltonian and a scattering manifold, an analogue of Knauf's degree can be defined.
\end{remark}


\section{Scattering in integrable systems} \label{sec/scattering_in_integrable_systems}

The goal of the present section is to recast the above theory of scattering in the context of Liouville integrability. The approach developed in the present section
will be applied to the Euler problem in Section~\ref{sec/scattering_2cp}.

\subsection{Reference systems}

 As we have seen in Section~\ref{sec/scattering}, reference systems can be 
used to define a scattering map, which is a map
between the asymptotic states at $t = - \infty$ and $t = + \infty$ of a given invariant manifold. 
For integrable systems, natural invariant manifolds are the fibers of
the corresponding integral map
$F$ and various unions of these fibers. It is thus natural to require that  the 
flow of a reference Hamiltonian maps
the set of asymptotic states  of a given fiber of $F$  to the set of asymptotic states of the same fiber. 
This leads to the following definition.
\begin{definition} \label{definition/reference}
Consider a scattering Hamiltonian $H$ which gives rise to an integrable  system $F \colon P \to \mathbb R^n$. A scattering Hamiltonian $H_{r}$ will be called a 
\textit{reference} Hamiltonian for this system if 
 $$
 F\left(\lim\limits_{t\to+\infty} g^t_{H_{r}}(x)\right) = F\left(\lim\limits_{t\to-\infty}g^t_{H_{r}}(x)\right)
 $$
 for every scattering trajectory $t \mapsto g^t_{H_{r}}(x)$.
\end{definition}

\begin{remark}
We note that Definition~\ref{definition/reference} can be generalized to the case of scattering and integrable systems defined on abstract symplectic manifolds. However,
for the purpose of the present paper it is sufficient to assume that $H$ and $H_r$ are as in Section~\ref{sec/scattering}. 
 \end{remark}

\begin{remark}
 In scattering theory it is usually required that the 
flow of a reference Hamiltonian maps 
the set of asymptotic states  of a given energy level to itself, which is a less restrictive assumption. Our point of view is that for integrable systems conserved quantities, such as the angular momentum,
should also be taken into account. 
\end{remark}

A series of examples of reference Hamiltonians in the above sense is given by rotationally symmetric potentials. This follows from the conservation of 
angular momentum. 
Another example is the Euler problem. We recall that the Hamiltonian of this problem is given by
$$H = \frac{\|p\|^2}{2} - \dfrac{\mu_1}{r_1} - \dfrac{\mu_2}{r_2}.$$
Let
$F = (H, L_z, G) \colon P \to \mathbb R^3$ be the integral map defined in Section~\ref{sec/preliminaries}. We have the following result. 

\begin{theorem} \label{theorem/ref}
Among all Kepler Hamiltonians only
$$H_{r_1} = \frac{1}{2}p^2 - \dfrac{\mu_1-\mu_2}{r_1} \ \mbox{ and } \ H_{r_2} = \frac{1}{2}p^2 - \dfrac{\mu_2-\mu_1}{r_2}$$
are reference Hamiltonians of the Euler problem $F = (H, L_z, G)$. In particular, the free Hamiltonian is a reference Hamiltonian of the Euler problem only in the case $\mu_1 = \mu_2$. 
\end{theorem}
\begin{proof}
 See Appendix~\ref{appendix/proof_reference_systems}.
\end{proof}

\begin{remark}
It follows from Theorem~\ref{theorem/ref} that a Kepler Hamiltonian with the
 strength $\mu_1 + \mu_2$  is not a reference of $F = (H,L_z,G)$, no matter where the center of attraction, resp., repulsion, is located. 
 For the strength $\mu_1 + \mu_2$ and only for this strength, the difference between the potentials is short-range. This implies that the \textit{M{\o}ller}  \textit{transformations} (or
 the \textit{wave} \textit{transformations})
 \cite{Knauf2011, Derezinski2013} are not defined 
 with respect to the reference Hamiltonians $H_{r_i}$, unless the reference
 flow is appropriately modified. 
 We note that the existence of M{\o}ller  transformations is important for the study of quantum scattering in this problem. 
\end{remark}

\subsection{Scattering invariants} \label{subsection/scattering_invariants}

Consider the Liouville fibration $F \colon s \to \mathbb R^n$.  Let $H_r$ be a reference Hamiltonian for $F$ such that $A^{\pm}(s) \subset A^{\pm}( s_r)$ holds. Setting $R = s$, we get the scattering map 
$$
S \colon B \to B, \ B = R / g^t_H.
$$
The scattering map $S$ allows to identify the asymptotic states of $s$ at $t = + \infty$ with the asymptotic states at $t = -\infty$. This results in a new total space 
$s_c$. We observe that under this identification the
asymptotic states of a given fiber of $F \colon s \to \mathbb R^n$ are mapped to the asymptotic states of the same fiber. This implies that $s_c$ is naturally fibered by $F$.
The resulting fibration will be denoted by
$$
F_c \colon s_c \to \mathbb R^n.
$$
We note that the invariants of the fibration $F_c$
contain essential information about the scattering dynamics. One such invariant is \textit{scattering monodromy} which we define as follows.

\begin{definition} \label{definition/scattering_monodromy}
Assume that 
$$F_c \colon s_c \to \mathbb R^n$$ is a torus bundle. The (usual) monodromy of this torus bundle will be called \textit{scattering monodromy} of the fibration $F$.
\end{definition}


\begin{remark}
We note that scattering monodromy in the above sense is related to \textit{non-compact monodromy} introduced in \cite{Efstathiou2017} for unbound systems with focus-focus singularities. 
It is known that focus-focus singularities come with a circle action \cite{Zung1997}. One can use this (global) action to compactify the fibration $F$ near a focus-focus fiber.
\end{remark}

\subsection{Planar potential scattering} \label{subsection/planar_potential_scattering} Here we shall discuss  the case $n = 2$ of planar scattering systems. The goal is to relate our notion 
of scattering monodromy
to the existing definition in terms of the deflection angle \cite{Bates2007, DullinWaalkens2008} and to make an explicit connection to the  scattering map.

Assume that  $V$ and $V_r$ are rotationally symmetric, that is, 
$$V(q) = W(\|q\|) \ \mbox{ and } \ V_r(q) = W_r(\|q\|) \  \mbox{ for some } \ W, W_r  \colon \mathbb R_{+} \to \mathbb R.$$ 
Then the angular momentum $L_z = xp_y - yp_x$ is conserved. Let $F = (H,L_z)$ be the integral map of the original system and $N$ be an arbitrary submanifold of the non-trapping set
 \begin{equation} \label{NTdefinition_planar}
  NT  = \{(h,l) \in \textup{image}(F) \mid F^{-1}(h,l) \subset s\}.
 \end{equation}
The  manifold $F^{-1}(N)$ is an invariant submanifold of the phase space $P$, which contains no trapping states (it consist of scattering states only).

Consider the case when $N = \gamma$ is a regular simple closed curve in $NT$. Let $R = F^{-1}(\gamma)$ and $S \colon B \to B,$ $B = F^{-1}(\gamma) / g^t_H$,
denote the corresponding scattering map. Then we have the following result.
\begin{theorem} \label{theorem/eq_definitions}
 The following statements are equivalent.
 
 \begin{itemize}
  \item[(1)] The scattering monodromy along $\gamma$ is a Dehn twist  of index $m$;
  \item[(2)] The variation of the \textit{deflection angle} along $\gamma$ equals $2\pi m$; 
  \item[(3)] The scattering map $S$ is a Dehn twist  of index $m$.
 \end{itemize} 
 \end{theorem}

 \begin{remark}
 By a Dehn twist of index $m$ we mean a homeomorphism of a $2$-torus such that its push-forward map is given by (the conjugacy class of) the matrix 
 $$ M = \begin{pmatrix} 1 & m \\ 0 & 1 
 \end{pmatrix}.
 $$
We note that the scattering manifold $B$ is  a $2$-torus in this case.
 \end{remark}
 \begin{remark}
 The \textit{total deflection angle} of a trajectory $g^t_H(x) = (q(t),p(t))$ is defined by 
$$\Phi = \int\limits_{-\infty}^{+\infty} \dfrac{d\varphi(q(t))}{dt} dt,$$
where $\varphi$ is the polar angle in the configuration $xy$-plane. The \textit{deflection angle} is defined as the difference of the total deflection angles for the original 
and the reference trajectories.
We note that $(2)$ is essentially the definition of scattering monodromy due to \cite{Bates2007, DullinWaalkens2008}.
 \end{remark}

\begin{proof}$(1) \Leftrightarrow (2).$ Let $(a,b)$ be homology cycles on the fiber $F_c^{-1}(\gamma(t_0))$ such that $b$ corresponds to the circle action given by  $L_z$. Transporting the cycles 
 along $\gamma$ we get $b \mapsto b$ and 
$a \mapsto a+mb$ for some integer $m$. But the difference
$$\Phi - \Phi_r = \int\limits_{-\infty}^{+\infty} \dfrac{d\varphi(q(t))}{dt} dt - \int\limits_{-\infty}^{+\infty} \dfrac{d\varphi(q_r(t))}{dt} dt,$$
where $g^t_{H_r} = (q_r(t), p_r(t))$ is a reference trajectory with the same energy and angular momentum,
can be seen as the  rotation number on the fibers of $F_c$. It follows that the variation of $\Phi - \Phi_r$
along $\gamma$ equals $2\pi m$.

$(2) \Leftrightarrow (3).$ The scattering map $S$ allows one to consider the compactified torus bundle
 $${\Pr} \colon {F^{-1}(\gamma)}^c \to S^1 = \mathbb R \cup \{\infty\},$$ where $\mathbb R$  corresponds to the time. The torus bundle considered in $(1)$ has the same total space,
 but is fibered over $\gamma$. Suppose that the monodromy of this bundle is given by the matrix 
  $$ M = \begin{pmatrix} 1 & m \\ 0 & 1 
 \end{pmatrix}.
 $$
 Then the monodromy of ${\Pr} \colon {F^{-1}(\gamma)}^c \to S^1$ is the same, for otherwise the total spaces would be different. 
 The result follows.
\end{proof}

\begin{remark}
 We note that in the original definition of \cite{DullinWaalkens2008} the potential $V$ is assumed to be repulsive and $V_r = 0.$ In this case the equivalence 
 $(1) \Leftrightarrow (2)$ follows from the results of \cite{Efstathiou2017}. 
\end{remark}

Theorem~\ref{theorem/eq_definitions}  gives three alternative definitions of monodromy in the case of scattering integrable systems in the plane. We observe that for the original definition
in terms of the deflection angle (Definition $(2)$) it is important that the scattering takes plane in the plane. On the other hand, from Section~\ref{sec/scattering} and the present section it follows that Definitions
$(1)$ and $(3)$ are suitable for scattering integrable
systems with many degrees of freedom, such as the Euler problem. Definition $(3)$, similarly to Knauf's degree, can be naturally applied to scattering systems even without integrability.

\section{Scattering in the Euler problem} \label{sec/scattering_2cp}

In this section we study scattering in the Euler problem using the reference Kepler Hamiltonians identified in the previous section. 
We will show that the Euler problem has non-trivial scattering monodromy of two different kinds: purely scattering monodromy and another kind, where both scattering and 
Hamiltonian monodromy are non-trivial.  
The latter kind can be observed only if the number of degrees of freedom $n \ge 3$. Purely Hamiltonian monodromy is also present in the problem; it survives the limiting cases of vanishing $\mu_i$, including the free flow. Scattering monodromy (of both kinds)
is trivial for the free flow. However, scattering monodromy of the second kind
is still present in the Kepler problem.

\subsection{Scattering map}
Let  $F = (H,L_z,G)$ denote the integral map of the Euler problem. 
Let $N$ be a submanifold of
 \begin{equation} \label{NTdefinition}
  NT  = \{(h,l,g) \in \textup{image}(F) \mid F^{-1}(h,l,g) \subset s\}.
 \end{equation}
The  manifold $F^{-1}(N)$ is an invariant submanifold of the phase space $P$, which contains scattering states only. Following the construction in 
Sections \ref{sec/scattering} and \ref{sec/scattering_in_integrable_systems},
we can define the scattering maps $S \colon B \to B$ with respect to $H$, the reference Kepler Hamiltonian $H_r = H_{r_1}$ or $H_r = H_{r_2},$ where
$$H_{r_1} = \frac{1}{2}p^2 - \dfrac{\mu_1-\mu_2}{r_1} \ \mbox{  and } \ H_{r_2} = \frac{1}{2}p^2 - \dfrac{\mu_2-\mu_1}{r_2},$$
and $B = F^{-1}(N) / g_H^t$ as in Subsection~\ref{subsection/scattering_map}.

\begin{remark}
 We recall that the scattering map $S$ is defined by
  $$
 S = (A^{-})^{-1} \circ A^{-}_{r} \circ (A^{+}_{r})^{-1} \circ A^{+},
 $$
 where 
 $$
A^{\pm} = (\hat{p}^{\pm},q_{\bot}^{\pm}) \colon  s^{\pm} / g^t_H \to  AS
\ \mbox{ and } \
A_r^{\pm} = (\hat{p}^{\pm},q_{\bot}^{\pm}) \colon  s_{r}^{\pm} / g^t_H \to  AS
$$
map  $s^{\pm} \subset P$ and $s_{r}^{\pm}$  to the asymptotic states $AS$. Here the index $r$ refers to 
a reference system ($H_{r_1}$ or $H_{r_2}$ in our case). 
\end{remark}

\begin{remark}We note that the potential
$$V = - \dfrac{\mu_1}{r_1} - \dfrac{\mu_2}{r_2}$$
of the Euler problem 
is short-range relative to $\widetilde{V}(q)= - (\mu_1+\mu_2)/\|q\|$, which is a Kepler potential. The reference potentials are Kepler potentials and are therefore rotationally symmetric.
It follows that the decay assumptions of Subsection~\ref{subsection/decay_assumptions} are met.
\end{remark}

\subsection{Scattering monodromy} \label{subsection/scattering_monodromy}
First we consider the case of a gravitational problem ($0 < \mu_2 < \mu_1$) with $H_r = H_{r_2}$ as the reference Kepler Hamiltonian. The other cases can be treated similarly; see
Subsection~\ref{subsection/general_case}.

For sufficiently large $h_0$ the $h = h_0$ slice of the bifurcation diagram has the form shown in 
Fig.~\ref{bifd_spatial_attr_gammas}. 
\begin{figure}[ht]
\begin{center}
\includegraphics[width=0.93\linewidth]{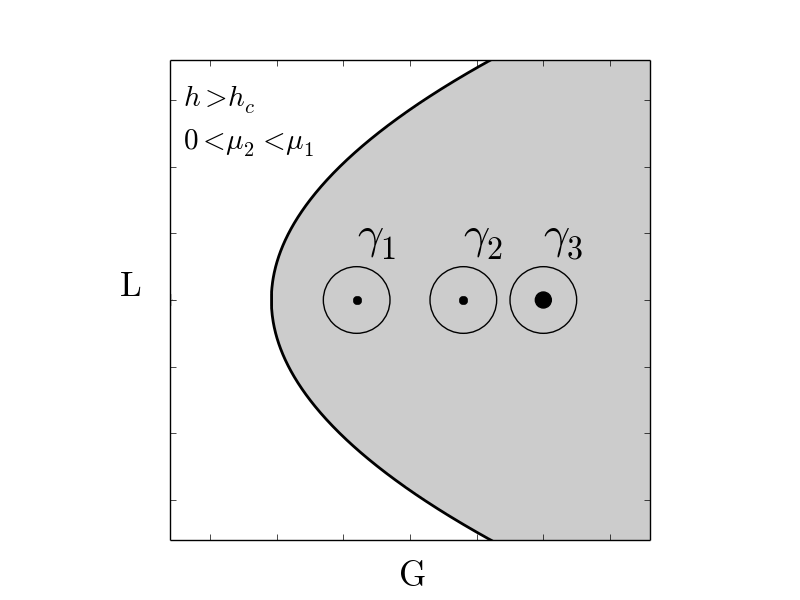}
\end{center}
\caption{Energy slice of the bifurcation diagram for the spatial Euler problem, attractive case. } 
\label{bifd_spatial_attr_gammas}
\end{figure}
Let $\gamma_i, i = 1, 2,3,$ be a simple closed curve in $$NT_{h_0} = \{(h,g,l) \in NT \mid h = h_0\}$$ that encircles the critical line $\ell_i$, where
\begin{multline*} 
\ell_1 = \{g = h+(\mu_2-\mu_1), \ l = 0\}, \  \ell_2 = \{g = h+(\mu_1-\mu_2), \ l = 0\} \ \mbox{ and } \ \\  \ell_3 = \{g = h+(\mu_1+\mu_2), \ l = 0\}.
\end{multline*}
For each $\gamma_i$, consider the torus bundle $F_{i} \colon E_{i} \to \gamma_i$, where the total space $E_i$ is obtained by gluing the ends of the fibers of $F$ over $\gamma_i$ via the scattering map $S$. 
We recall that \textit{scattering monodromy} along $\gamma_i$ with respect to $H_r$ is defined as the usual monodromy of the torus 
bundle $F_{i} \colon E_{i} \to \gamma_i$; see Definition~\ref{definition/scattering_monodromy} and Appendix~\ref{appendix/Hamiltonian_monodromy}.

\begin{remark} \label{remark/alternative_definition}
 Alternatively, one can define $F_{i} \colon E_{i} \to \gamma_i$ by gluing the fibers of the original and the reference integral maps at infinity. Both definitions are equivalent in the sense
 that the monodromy of the resulting torus bundles are the same.
\end{remark}

Consider a starting point $\gamma_i(t_0) \in \gamma_i$ in the region where  $l > 0$. 
We choose a basis $(c_\xi,  c_\eta, c_\varphi)$ of the first homology group $\textup{H}_1(F_{i}^{-1}(\gamma_i(t_0))) \simeq \mathbb Z^3$ as follows.
The cycle $c_\xi = c^{o}_\xi \cup c^{r}_\xi$ is obtained by gluing the non-compact $\xi$-coordinate lines $c^{o}_\xi$ for the original and $c^{r}_\xi$ for the reference systems at infinity. In other words,
for
we glue the lines
$$p^2_\xi = \dfrac{(\xi^2 - 1)(2h\xi^2 + 2(\mu_1+\mu_2)\xi - 2g) - l^2}{(\xi^2-1)^2}$$
on  $F^{-1}(\gamma_i(t_0)), \gamma_i(t_0) = (h,g,l)$,
and 
$$p^2_\xi = \dfrac{(\xi^2 - 1)(2h\xi^2 + 2(\mu_2-\mu_1)\xi - 2g) - l^2}{(\xi^2-1)^2}$$
on the reference fiber $F_r^{-1}(\gamma_i(t_0))$  at the limit points $\xi = \infty$, $p_\xi = \pm \sqrt{2h}.$ 
The cycles  $c_\eta$ and $c_\varphi$  are such that their projections onto the configuration space coincide with coordinate lines of $\eta$ and $\varphi$, respectively. In other words, 
the cycle $c_\eta$ on $F^{-1}(\gamma_i(t_0))$ is given by
$$p^2_\eta = \dfrac{(1-\eta^2)(-2h\eta^2 - 2(\mu_1-\mu_2)\eta - 2g) - l^2}{(1-\eta^2)^2}$$
and  $c_\varphi$ is an orbit of the circle action given by the Hamiltonian flow of the momentum $L_z$.
We have the following result.

\begin{theorem} \label{theorem/main}
 The monodromy matrices $M_i$ of $E_i \to \gamma_i$ with respect to the natural basis $(c_\xi, c_\eta, c_\varphi)$ have the  form 
 $$  
  M_1  = \begin{pmatrix}
  1 & 0 & 0 \\
  0 & 1 & 1 \\
  0 & 0 & 1
 \end{pmatrix}, \  M_2  = \begin{pmatrix}
  1 & 0 & -1 \\
  0 & 1 & 1 \\
  0 & 0 & 1
 \end{pmatrix} \ \mbox{ and }  \
  M_3  = \begin{pmatrix}
  1 & 0 & 1 \\
  0 & 1 & 0 \\
  0 & 0 & 1
 \end{pmatrix}.
 $$
\end{theorem}

\begin{proof}
{\bf Case 1}, loop $\gamma_1$. First we note that the cycle $c_\varphi$ is preserved under the parallel transport along $\gamma_1$. This follows from the fact that 
$L_z$ generates a free fiber-preserving circle action on $E_i$. The cycles 
$c_\xi$ and $c_\eta$ can be naturally transported only in the regions where $l \ne 0$. We thus need to understand what happens at the critical plane $l = 0$. 

Let  $R > 1$ be a sufficiently large number. Then 
$$E_{1,R} = \{x \in E_1 \mid \xi(x) > R\}$$
has exactly two connected components, which we denote by $E^+_{1,R}$ and $E^-_{1,R}$.
We define a $1$-form $\alpha$ on (a part of) $E_i$ by the formula
$$
\alpha = pdq - \chi(\xi)p_\xi(h,g,l,\xi)d\xi,
$$
where $\chi(\xi)$ is a bump function such that 

(i) $\chi(\xi) = 0$ when $\xi < R;$ 

(ii) $\chi(\xi) =1$ when $\xi > 1 + R.$ \\ 
The square root function $p_\xi(h,g,l,\xi)$ is assumed to be positive on $E^+_{1,R}$ and negative on $E^-_{1,R}$. 
By construction, the $1$-form
$\alpha$ is well-defined and smooth on $E_i$ outside collision points. Since 
$$d\alpha = dp \wedge dq = - \omega \ \mbox{ on } \  F^{-1}(\gamma_i) \cup F_r^{-1}(\gamma_i) \subset E_i,$$
we have that
$d\alpha = 0$ on each fiber of $F_i$.

Consider the modified actions with respect to the form $\alpha$:
$$I_\varphi = \dfrac{1}{2\pi}\int\limits_{c_\varphi} \alpha, \ I_\eta = \dfrac{1}{2\pi}\int\limits_{c_\eta} \alpha \ \ \mbox{ and } \ \
I^{mod}_\xi = \dfrac{1}{2\pi}  \int\limits_{c_\xi} \alpha.$$
The modified actions are well defined and, in view of $d\alpha = 0$, depend only on the homology classes of $c_\xi,c_\eta$ and $c_\varphi.$ It follows that
$I_\varphi$ and $I_\eta$ coincide with the `natural' actions (defined as the integrals over the usual $1$-form  $pdq$). We note that the `natural' $\xi$-action 
$$I_\xi = \dfrac{1}{2\pi}  \int\limits_{c_\xi} pdq$$
diverges, cf. \cite{DullinWaalkens2008}.
From the continuity of the modified actions at $l = 0$ it follows that the corresponding scattering monodromy matrix has the form
 $$  
  M_{1}  = \begin{pmatrix}
  1 & 0 & m_{1} \\
  0 & 1 & m_{2} \\
  0 & 0 & 1
 \end{pmatrix}.
 $$
Since the modified actions do not have to be smooth at $l = 0$, the integers $m_1$ and $m_2$ are not necessarily zero.
 In order to compute these integers  we need to compare the
 derivatives 
$\partial_l I_\eta$ and $\partial_l I_\xi$  at $l \to \pm 0.$ 
 A computation of the corresponding 
residues gives
\begin{equation*}
 \lim_{\l \to {\pm 0}} \partial_l I_\eta = \lim_{\l \to {\pm 0}} \dfrac{1}{2\pi}  \partial_l\int\limits_{c_\eta} pdq  = \begin{cases}
0,  & \mbox{ when }  \ g < h + \mu_2-\mu_1,\\
\mp 1/2, & \mbox{ when } \ \mu_2-\mu_1< g - h < \mu_1-\mu_2,
\end{cases}
\end{equation*}
and 
\begin{multline*} 
 \lim_{\l \to {\pm 0}} \partial_lI^{mod}_\xi = \lim_{\l \to {\pm 0}} \left(\dfrac{1}{2\pi}  \partial_l\int\limits_{c^{o}_\xi} pdq - \dfrac{1}{2\pi}  \partial_l\int\limits_{c^{r}_\xi} pdq\right) - \\
 \lim_{\l \to {\pm 0}} \dfrac{1}{2\pi} \int\limits_{c_\xi} \chi(\xi)p_\xi(h,g,l,\xi)d\xi = 0
\end{multline*}
(for the two ranges of $g$). It follows that $m_1 = 0$ and $m_2 = 1$. 


{\bf Case 2},  loop $\gamma_2$. This case is similar to {\bf Case 1}. The corresponding limits are given by
\begin{equation*}
 \lim_{\l \to {\pm 0}} (\partial_l I_\eta, \partial_lI^{mod}_\xi)  = \begin{cases}
(\mp 1/2, 0),  & \mbox{ when } \  \mu_2-\mu_1< g - h < \mu_1-\mu_2,\\
(\mp 1, \pm 1/2), & \mbox{ when } \ \mu_1-\mu_2< g - h < \mu_1+\mu_1.
\end{cases}
\end{equation*}

{\bf Case 3},  loop $\gamma_3$. The computation in this case is also similar to {\bf Case 1}. The corresponding limits are given by
\begin{equation*}
 \lim_{\l \to {\pm 0}} (\partial_l I_\eta, \partial_lI^{mod}_\xi)  = \begin{cases}
(\mp 1, \pm 1/2),  & \mbox{ when } \  \mu_1-\mu_2< g - h < \mu_1+\mu_2,\\
(\mp 1, 0), & \mbox{ when } \ h + \mu_1+\mu_2< g.
\end{cases}
\end{equation*}
\end{proof}

\begin{remark}
 One difference between {\bf Case 3} and the other cases is the topology of the critical fiber, around which scattering monodromy is defined. In {\bf Case 3} the critical
 fiber is the product of a pinched cylinder and a circle, whereas in the other cases 
 it is the product  of a pinched torus and a real line. 
 This implies, in fact, that {\bf Case 3} is purely scattering, whereas in the other cases Hamiltonian monodromy is present; see Subsection~\ref{sec/topology} for details.
 \end{remark}
 
 \renewcommand{\proofname}{Proof for Case 3 of Theorem~\ref{theorem/main}}
 
 \begin{remark}
 Theorem~\ref{theorem/main} admits the following geometric proof in the purely scattering case.
 \begin{proof}
 The action 
$$
 I'_\eta = \left.
  \begin{cases}
    I_\eta, & \text{if } l \ge 0 \\
    I_\eta - 2l , & \text{if } l < 0.
  \end{cases} \right.
  $$
  is smooth and globally defined (over $\gamma_3$). Moreover, the corresponding  circle action extends to a free action in $F_3^{-1}(D_3)$, where 
$D_3 \subset NT_{h_0}$ is a $2$-disk such that $\partial D_3 = \gamma_3$. Since there is also a circle action given by $I_{\varphi}$, 
the result can be also deduced from the general theory developed in \cite{EfstathiouMartynchuk2017, Martynchuk2017}.
\end{proof}
We note that from the last proof it follows that the choice of a reference Kepler Hamiltonian
does not affect the result in the purely scattering case. This agrees with the point of view presented recently in \cite{Efstathiou2017} for two degree of freedom systems
with focus-focus singularities. For the curves $\gamma_1$ and $\gamma_2$,  it is important which of the two 
reference Kepler Hamiltonians is chosen; see Subsection~\ref{subsection/general_case}. 
\end{remark}

\renewcommand{\proofname}{Proof}
As a corollary, we get the following result for the scattering map in the purely scattering case  of the curve $\gamma_3$.
\begin{theorem} \label{theorem/dehntwist}
 The scattering map $S \colon B_3 \to B_3,$ where $B_3 = F^{-1}(\gamma_3)/g^t_H,$ is a Dehn twist. The push-forward map is conjugate in
 $SL(3,\mathbb Z)$ to 
 $$
 S_{\star} =
  \begin{pmatrix}
    1 & 0 & 1 \\
    0 & 1 & 0 \\
    0 & 0 & 1
  \end{pmatrix}.
 $$
\end{theorem}
\begin{proof}
 The proof is similar to the proof of the equivalence $(2) \Leftrightarrow (3)$ given in Theorem~\ref{theorem/eq_definitions}. The scattering map $S$ allows one to consider the compactified torus bundle
 $${\Pr} \colon {F^{-1}(\gamma_3)}^c \to S^1 = \mathbb R \cup \{\infty\},$$ where $\mathbb R$  corresponds to the time. The torus bundle
 $F_3 \colon E_3 \to \gamma_3$ has the same total space,
 but is fibered over $\gamma_3$. 
 By Theorem~\ref{theorem/main}, the monodromy of the bundle $F_3 \colon E_3 \to \gamma_3$ is given by the matrix 
  $$ M = \begin{pmatrix} 1 & 0 & 1 \\ 0 & 1 & 0 \\ 0 & 0 & 1 
 \end{pmatrix}.
 $$
 Then the monodromy of the first bundle ${\Pr} \colon {F^{-1}(\gamma_3)}^c \to S^1$ is the same, for otherwise the total spaces would be different. 
 The result follows.
\end{proof}

\begin{remark}
It follows from the proof and Subsection~\ref{subsection/general_case} that Theorem~\ref{theorem/dehntwist} holds for any $\mu_i \ne 0$ and
for any regular closed curve $\gamma \subset NT$ such that 
\begin{enumerate}
\item[1.] The energy value $h$ is positive on $\gamma$;
\item[2.] $\gamma$ encircles the critical line $\{g = h+\mu_1+\mu_2, \ l = 0\}$ exactly once and does not encircle any other  line of critical values;
\item[3.] $\gamma$ does not cross critical values of $F$.
\end{enumerate}
\begin{figure}[ht]
\begin{center}
\includegraphics[width=0.93\linewidth]{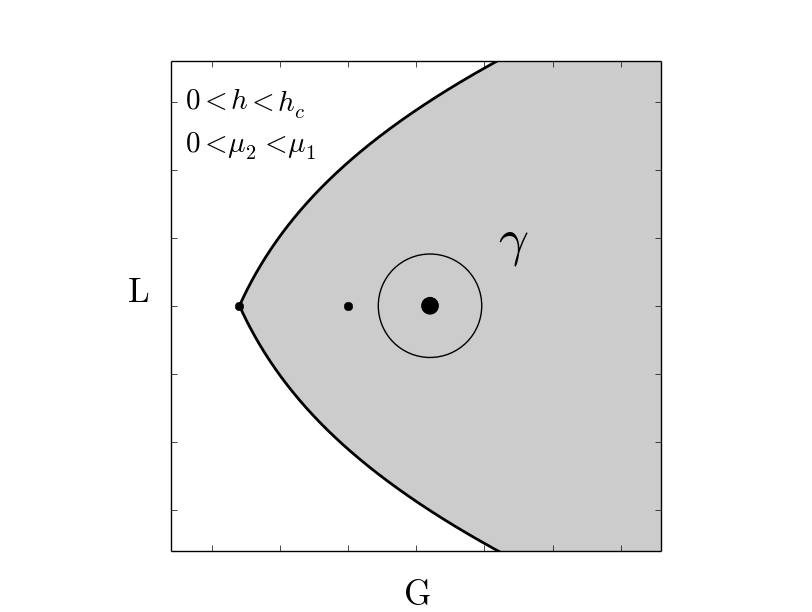}
\end{center}
\caption{Energy slice of the bifurcation diagram for the spatial Euler problem, attractive case.  } 
\label{bifd_spatial_attr_gamma}
\end{figure}
 It can be shown that such a curve $\gamma$ always exists; an example is given in Fig.~\ref{bifd_spatial_attr_gamma}. 
We note that the third condition can be weakened in the case  $-\mu_1 < \mu_2 < 0$. 
In this case the attraction of $\mu_1$ dominates the repulsion of $\mu_2$ and, as a result, 
bound motion coexists with unbound motion for a range of positive energies. Instead of $F^{-1}(\gamma)$ one may consider its unbounded component.
\end{remark}


\subsection{Topology} \label{sec/topology}

As we have noted before, alongside scattering monodromy, the Euler problem admits also another type of invariant --- Hamiltonian monodromy. Here we consider the generic case of
$|\mu_1| \ne |\mu_2| \ne 0$ in the case of positive energies. The case of negative energies 
is similar --- it has been discussed in detail in \cite{WaalkensDullinRichter04}. The critical cases can be easily computed from the generic case by considering curves that encircle
more than one of the  singular lines 
\begin{multline*}
\ell_1 = \{g = h+(\mu_2-\mu_1), \ l = 0\}, \  \ell_2 = \{g = h+(\mu_1-\mu_2), \ l = 0\} \ \mbox{ and } \ \\  \ell_3 = \{g = h+(\mu_1+\mu_2), \ l = 0\}.
\end{multline*}
Let $\gamma_i$ be  a closed curve   that encircles only the critical line $\ell_i$; see Fig.~\ref{bifd_spatial_attr_gamma}. The fibration $F \colon F^{-1}(\gamma_i) \to \gamma_i$
is a  $T^2 \times \mathbb R$-bundle. 
The following theorem shows that  the Hamiltonian monodromy (see Appendix~\ref{appendix/Hamiltonian_monodromy}) is non-trivial  along the curves 
$\gamma_1$ and $\gamma_2$ 
and is trivial along $\gamma_3$.
\begin{theorem}
 The Hamiltonian monodromy of $F \colon F^{-1}(\gamma_i) \to \gamma_i, \ i = 1,2$, is conjugate in $SL(2,\mathbb Z) \subset SL(3,\mathbb Z)$ to 
$$
 M =
  \begin{pmatrix}
    1 & 0 & 0 \\
    0 & 1 & 1 \\
    0 & 0 & 1
  \end{pmatrix}.
 $$
 Here the right-bottom $2 \times 2$ block acts on $T^2$ and the left-top $1 \times 1$ block acts on $\mathbb R$.
 \end{theorem}
 \begin{proof}
 The result follows from the proof of Theorem~\ref{theorem/main}. For completeness, we give an independent proof below.
 
  After the reduction of the surface $H^{-1}(h)$ with respect to the flow $g^t_H$ we get a singular $T^2$ 
torus fibration over a disk $D_i, \ \partial D_i = \gamma_i,$ with exactly one focus-focus point.
The result then follows 
from \cite{Lerman1994, Matveev1996, Zung1997}. This argument applies to both of the  lines $\ell_1$ and $\ell_2$. Since the flow of 
$L_z$  gives a global circle action,
the monodromy matrix $M$ is the same in both cases; see
\cite{Cushman2002}. 
\end{proof}

\begin{theorem}
 The Hamiltonian monodromy of $F \colon F^{-1}(\gamma_3) \to \gamma_3$ is trivial.
 \end{theorem}
 \begin{proof}
Observe that the Hamiltonian flows of
$I_\varphi,$
$$
 I'_\eta = \left.
  \begin{cases}
    I_\eta, & \text{if } l \ge 0 \\
    I_\eta - 2l , & \text{if } l < 0.
  \end{cases} \right.
  $$
  and $H$ generate a global $\mathbb T^2 \times \mathbb R$ action on $F^{-1}(\gamma_3)$. 
  It follows that the bundle $F \colon F^{-1}(\gamma_3) \to \gamma_3$ is principal. Since $\gamma_3$ is a circle, it is also trivial. 
\end{proof}

We note that  Hamiltonian monodromy is an intrinsic invariant of the Euler problem, related to the non-trivial topology of the integral map $F$.
Interestingly, it is also present in the critical cases: 

(1) $\mu_1 = \mu_2$ (symmetric Euler problem) \cite{WaalkensDullinRichter04}, 

(2) $\mu_1$ or $\mu_2 =0$ (Kepler problem) \cite{DullinWaalkens2018} and 

(3) $\mu_1 = \mu_2 = 0$ (the free flow).\\
In the case of bound motion (1) and (2) are due to \cite{WaalkensDullinRichter04} and \cite{DullinWaalkens2018}, respectively. 
From the scattering perspective Hamiltonian monodromy is recovered if one considers the original Hamiltonian $H$ also as a reference.

\subsection{General case} \label{subsection/general_case}

Here we consider the case of of arbitrary strengths $\mu_i$. We observe that the
scattering monodromy matrices with respect to the reference Kepler Hamiltonians $H_{r_1}$ and $H_{r_2}$ are necessarily of the form
$$
 \left(\begin{array}{ccc}
    1 & 0 & m \\
    0 & 1 & n \\
    0 & 0 & 1
  \end{array}\right)
  $$
for some integers $m$ and $n$. These integers (for different choices of the strengths $\mu_i$ and the critical lines $\ell_i$) are given in Table~\ref{table}.

\begin{remark}
We note that one can compute the  monodromy matrices in the critical cases from 
the matrices found in the generic cases. Specifically, it is sufficient to consider the curves that encircle more than one critical line $\ell_i$ and multiply
the monodromy matrices found around each of these lines. For instance, the monodromy matrix around the curve $g = h$ in the free flow equals
the product of the three monodromy matrices found in (any) generic Euler problem.
\end{remark}

\begin{table}
\begin{center}
\begin{tabular}{|l|l|l|l|}
\hline
& $\gamma_1$ & $\gamma_2$  & $\gamma_3$ \\
\hline
& \multicolumn{3}{c|}{{\bf Scattering monodromy w.r.t. $H_{r_1}$}} \\
 \hline
 {\bf Generic} & \multicolumn{3}{c|}{ } \\
 \hline
$|\mu_1| \ne |\mu_2| \ne 0$  & $m = -1, n = 1$  & $m = 0, n = 1$ & $m = 1, n = 0$ \\
\hline
 {\bf Critical } & \multicolumn{3}{c|}{  } \\
 \hline
$-\mu_1 = \mu_2 < 0$ & $m = -1, n = 1$  & $m = 0, n = 1$ & $n = 1, n = 0$\\
\hline
$0< \mu_1 = \mu_2$ & \multicolumn{2}{c|}{$m = -1, n = 2$} & $m = 1, n = 0$\\
\hline
$\mu_1 = \mu_2 < 0$ & \multicolumn{2}{c|}{$m = -1, n = 2$} & $m = 1, n = 0$\\
\hline
$\mu_1 = \mu_2 = 0$ & \multicolumn{3}{c|}{$m = 0, n = 2$} \\
\hline
$0 = \mu_2 < \mu_1$ & $n = 1$  & $m = 0, n = 1$ &  $ \ \ \ \ \ \ \ \ \ \ \ \ m = 0,$\\ 
\hline
$\mu_1 < \mu_2 = 0$ & $m = -1, n = 1$  & \multicolumn{2}{c|}{$m = 1, n = 1$} \\
\hline
& \multicolumn{3}{c|}{{\bf Scattering monodromy w.r.t. $H_{r_2}$}} \\
 \hline
 {\bf Generic } & \multicolumn{3}{c|}{ } \\
 \hline
$|\mu_1| \ne |\mu_2| \ne 0$ & $m = 0, n = 1$  & $m = -1, n = 1$ & $m = 1, n = 0$ \\
\hline
 {\bf Critical } & \multicolumn{3}{c|}{  } \\
 \hline
$-\mu_1 = \mu_2 < 0$ & $m = 0, n = 1$  & $m = -1, n = 1$ & $m = 1, n = 0$\\
\hline
$0 = \mu_2 < \mu_1$ & $n = 1$  & $m = -1, n = 1$ &  $ \ \ \ \ \ \ \ \ \ \ \ \ m = 1,$\\
\hline
$\mu_1 < \mu_2 = 0$ & $m = 0, n = 1$  & \multicolumn{2}{c|}{$m = 0, n = 1$} \\
\hline
\end{tabular}
\end{center}
\caption{Scattering monodromy, general case.} \label{table}
\end{table}

\section{Discussion} \label{sec/discussion}

In the present paper we have shown that the spatial Euler problem, alongside non-trivial Hamiltonian monodromy \cite{WaalkensDullinRichter04},
has non-trivial scattering monodromy of two different types: pure and mixed scattering monodromy.
The first type reflects the presence of a special periodic orbit --- a collision orbit that bounces between the
two centers --- and the associated   trapping trajectories. In the spatial case one can go around these trajectories and compare
the flow at infinity to an appropriately chosen Kepler problem. 
Scattering monodromy of the second type is related to the difference in dynamics of the  original and the reference systems; here in addition to scattering monodromy also Hamiltonian monodromy is present. Interestingly,
scattering monodromy of the second type survives vanishing of one of the centers:
it can be also observed in the limiting case of attractive and repulsive Kepler problems
$$H_{r_1} = \frac{1}{2}p^2 - \dfrac{\mu}{r_1} \ \mbox{ and } \ H_{r_2} = \frac{1}{2}p^2 + \dfrac{\mu}{r_2}.$$
Hamiltonian monodromy is present not only in the Kepler problem \cite{DullinWaalkens2018}, but also in the free flow.
The purely scattering monodromy is special to the genuine Euler problem;
we conjecture that this invariant is also present in the restricted three-body
problem.

\section{Acknowledgements}

We would like to thank 
Prof. Dr. A. Knauf for the useful and stimulating discussions.

\appendix

\section{Hamiltonian monodromy} \label{appendix/Hamiltonian_monodromy}

Consider an integrable Hamiltonian system 
$
F = (F_1 = H, F_2, \ldots, F_n)
$
on a $2n$-dimensional symplectic manifold $(M, \omega)$. If the fibers of the integral map $F$ are \textit{compact} and \textit{connected}, then 
according to the classical
Arnol'd-Liouville theorem \cite{Arnold1968} a tubular neighborhood of each regular fiber
is
a trivial torus bundle $D^n \times T^n$ admitting action-angle coordinates. Hence
$$
F \colon F^{-1}(R) \to R,
$$
where $R \subset \textup{image}(F)$ is the set of regular values of $F$,
is a locally trivial torus bundle.
This bundle is, however, not necessary globally trivial even from the topological viewpoint.
One geometric invariant that measures this non-triviality  was  introduced by Duistermaat in \cite{Duistermaat1980} and is called \textit{Hamiltonian monodromy}. Specifically, 
Hamiltonian monodromy is defined as a representation
\begin{align*}
  \pi_1(R,\xi_0) \to \textup{Aut}\,H_1(F^{-1}(\xi_0)) \simeq \mathrm{GL}(n, \mathbb Z)
\end{align*}
of the fundamental group $\pi_1(R,\xi_0)$ in the group of automorphisms of the integer homology group $H_1(F^{-1}(\xi_0)) \simeq \mathbb Z^n$.
Each element $\gamma \in \pi_1(R,\xi_0)$ acts via parallel transport of integer homology cycles \cite{Duistermaat1980}.

Since the pioneering work of Duistermaat,
Hamiltonian monodromy and its quantum counterpart \cite{Cushman1988, Vu-Ngoc1999} have been observed in many integrable systems of physics and mechanics. General results are known that allow to compute this invariant
in specific examples. It has been shown in \cite{Lerman1994, Matveev1996, Zung1997} that in the typical case of $n = 2$ degrees of freedom non-trivial Hamiltonian monodromy is manifested by the presence of the 
so-called \textit{focus-focus} points of the map $F$. In the case of a  global circle action Hamiltonian monodromy (and, more generally,
\textit{fractional monodromy} \cite{Nekhoroshev2006}) can be computed in terms of the singularities of the circle action  \cite{EfstathiouMartynchuk2017, Martynchuk2017}.

\begin{remark} A notion of monodromy can be defined for torus bundles that do not necessarily come from an integrable system and also in the case of bundles with non-compact fibers  
(for instance, in the case of cylinder  bundles). Specifically, consider a bundle 
$F \colon F^{-1}(\gamma) \to \gamma, \ \gamma = S^1.$
It can be obtained from a direct product 
$[0,1] \times F^{-1}(\gamma(t_0))$ by gluing the boundaries via a non-trivial homeomorphism $f$, called the \textit{monodromy} of the bundle. We call this monodromy 
 \textit{Hamiltonian} if $F$ comes from a completely integrable system. In this case
the push-forward map $f_{\star}$
coincides with the 
automorphism given by the parallel transport.

We note that non-compact fibrations appear in the Euler problem in the case of positive energies and in various other integrable systems. We  mention the works
\cite{Flashka1988, Kudryavtseva2011, Martynchuk2015} and \cite{Bates2007, DullinWaalkens2008, Zagryadskii2012, Efstathiou2017}.
For systems that are both scattering and integrable scattering monodromy and Hamiltonian monodromy  coincide if the reference   is given by the original Hamiltonian $H$.

\end{remark}

\section{Bifurcation diagrams for the planar problem}  \label{appendix/bifurcation_diagrams_for_the_planar_problem}
In this section we give \textit{bifurcation diagrams} of the planar Euler problem in the case of arbitrary strengths $\mu_i$. The computation has been performed in 
Section~\ref{sec/bifurcationdiagrams}; more details can be found in \cite{Deprit1962, WaalkensDullinRichter04, Seri2015}.

\begin{figure}[ht]
	\begin{center}
	\vspace{-1mm}
		\includegraphics[width=\linewidth]{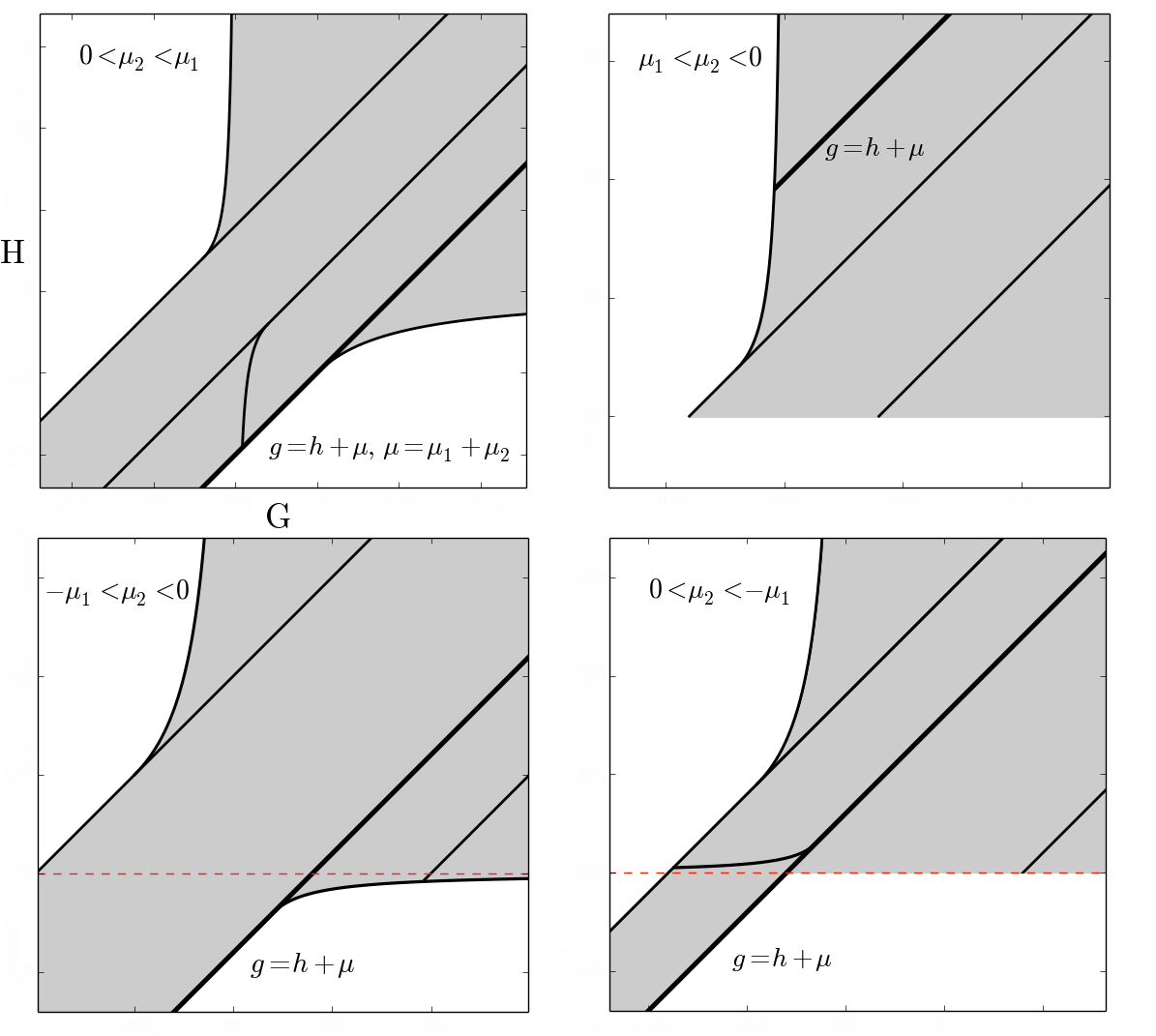}
	\end{center}
	\vspace{-3mm}
	\caption{Bifurcation diagrams for the planar problem, generic cases
	$|\mu_1| \ne |\mu_2| \ne 0$.  Top: attractive (left), repulsive (right). Bottom: mixed. } 
	\label{bifd_planar_generic_case}
\end{figure}
\begin{figure}[ht]
	\begin{center}
		\includegraphics[width=1\linewidth]{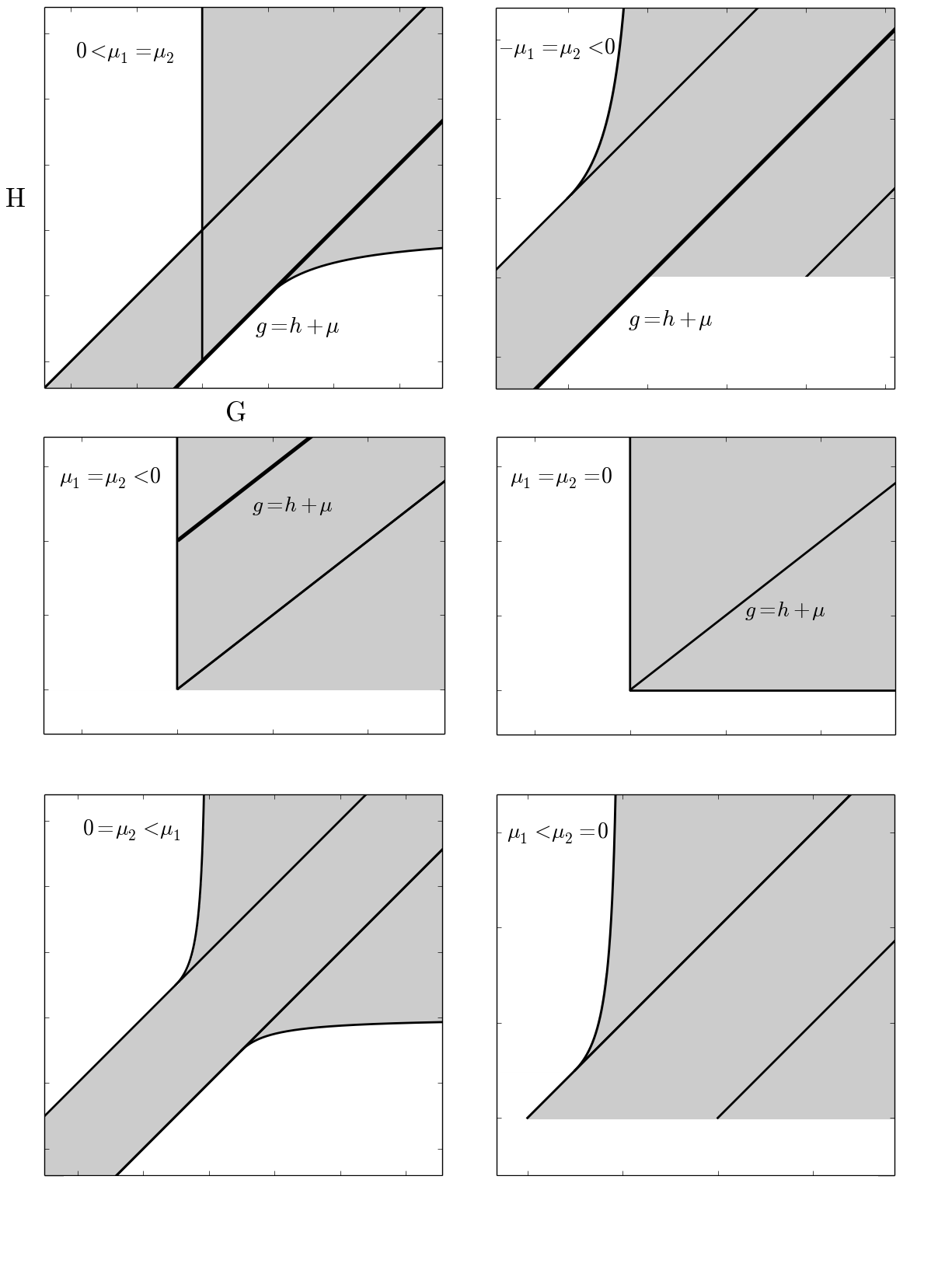}
	\end{center}
	\vspace{-0.32cm}
	\caption{Bifurcation diagrams for the planar problem, non-generic cases 
	$|\mu_1| = |\mu_2|$ or $\mu_1 \mu_2 = 0$. From left to right, from top to bottom:
	symmetric attractive, anti-symmetric, symmetric repulsive, free flow, attractive Kepler problem, repulsive Kepler problem.}
	\label{bifd_planar_critical_case}
\end{figure}

The computation of Section~\ref{sec/bifurcationdiagrams} yields the following critical
lines 
\begin{multline} \label{eq/critical_lines_planar}
\ell_1 = \{g = h+\mu_2-\mu_1\}, \ \ell_2 = \{g = h+\mu_1-\mu_2\} \\  \mbox{ and } \  \ell_3  = \{g = h+\mu\}, \ \mu = \mu_1+\mu_2,
\end{multline}
and the critical curves
\begin{align*}\{g &= \mu \cosh\lambda/2, \ h =  -\mu /2\cosh\lambda\}, \\
\{g &= (\mu_1 - \mu_2) \sin \nu/2, \ h = (\mu_2 - \mu_1)/2 \sin \nu\}.
\end{align*}
Points that do not correspond to any physical motion must be removed from the obtained set. 
The resulting diagrams are given in 
Figs.~\ref{bifd_planar_generic_case} and \ref{bifd_planar_critical_case}. 
Here
we distinguish two cases: generic case when the strengths
$|\mu_1| \ne |\mu_2| \ne 0$  and the remaining critical cases.

We note that the critical cases occur when $|\mu_1| = |\mu_2|$ or when $\mu_1 \mu_2 = 0$. In the case $\mu_1 = -\mu_2 \ne 0$ the attraction of one of the centers equalizes the repulsion of the other center, making the bifurcation diagram qualitatively different from 
the cases when $-\mu_1 < \mu_2 < 0$ or $0 < \mu_2 < - \mu_1$. However, we still have the three different critical lines 
$\ell_1, \ell_2$ and $\ell_3$. In the other critical cases collisions of the critical lines $\ell_i$ occur. For instance, 
$\mu_1 = 0$ implies that $\ell_1 = \ell_3$ and so on. The same situation takes place in the spatial problem.

\section{Proof of Theorem~\ref{theorem/ref}} \label{appendix/proof_reference_systems}
We shall show that the Euler  problem has two natural reference Hamiltonians when $\mu_1 \ne \mu_2$ and one otherwise.
\begin{theorem} 
Among all Kepler Hamiltonians only
$$H_{r_1} = \frac{1}{2}p^2 - \dfrac{\mu_1-\mu_2}{r_1} \ \mbox{ and } \ H_{r_2} = \frac{1}{2}p^2 - \dfrac{\mu_2-\mu_1}{r_2}$$
are reference Hamiltonians of $F = (H,L_z,G)$. In particular, the free Hamiltonian is a reference Hamiltonian of $F$ only in the case $\mu_1 = \mu_2$. 
\end{theorem}
\begin{proof} \textit{Sufficiency.} 
Consider the Hamiltonian $H_{r_1}$.  Let
 \begin{equation*}
G_{r_1} = H_{r_1} + \frac{1}{2}(L^2-a^2(p_x^2+p_y^2)) + a(z+a)\frac{\mu_1-\mu_2}{r_1}.
\end{equation*}
From Section~\ref{sec/separation_and_integrability} (see also Eq.~\eqref{eq/thirdintegral}) it follows that the functions
$H_{r_1}, L_z$ and $G_{r_1}$ Poisson commute. This implies that any trajectory
$g^t_{H_{r_1}}(x)$ belongs to the common level 
set of
$F_{r_1} = (H_{r_1}, L_z, G_{r_1}).$
For a scattering trajectory we thus get
 $$
 F_{r_1}\left(\lim\limits_{t\to+\infty} g^t_{H_{r_1}}(x)\right) = F_{r_1}\left(\lim\limits_{t\to-\infty}g^t_{H_{r_1}}(x)\right).
 $$
A straightforward computation of the limit shows that also
$$
F\left(\lim\limits_{t\to+\infty} g^t_{H_{r_1}}(x)\right) = F\left(\lim\limits_{t\to-\infty} g^t_{H_{r_1}}(x)\right).
$$
The case of $H_{r_2}$ is completely analogous.

\begin{figure}[ht]
\begin{center}
\includegraphics[width=\linewidth]{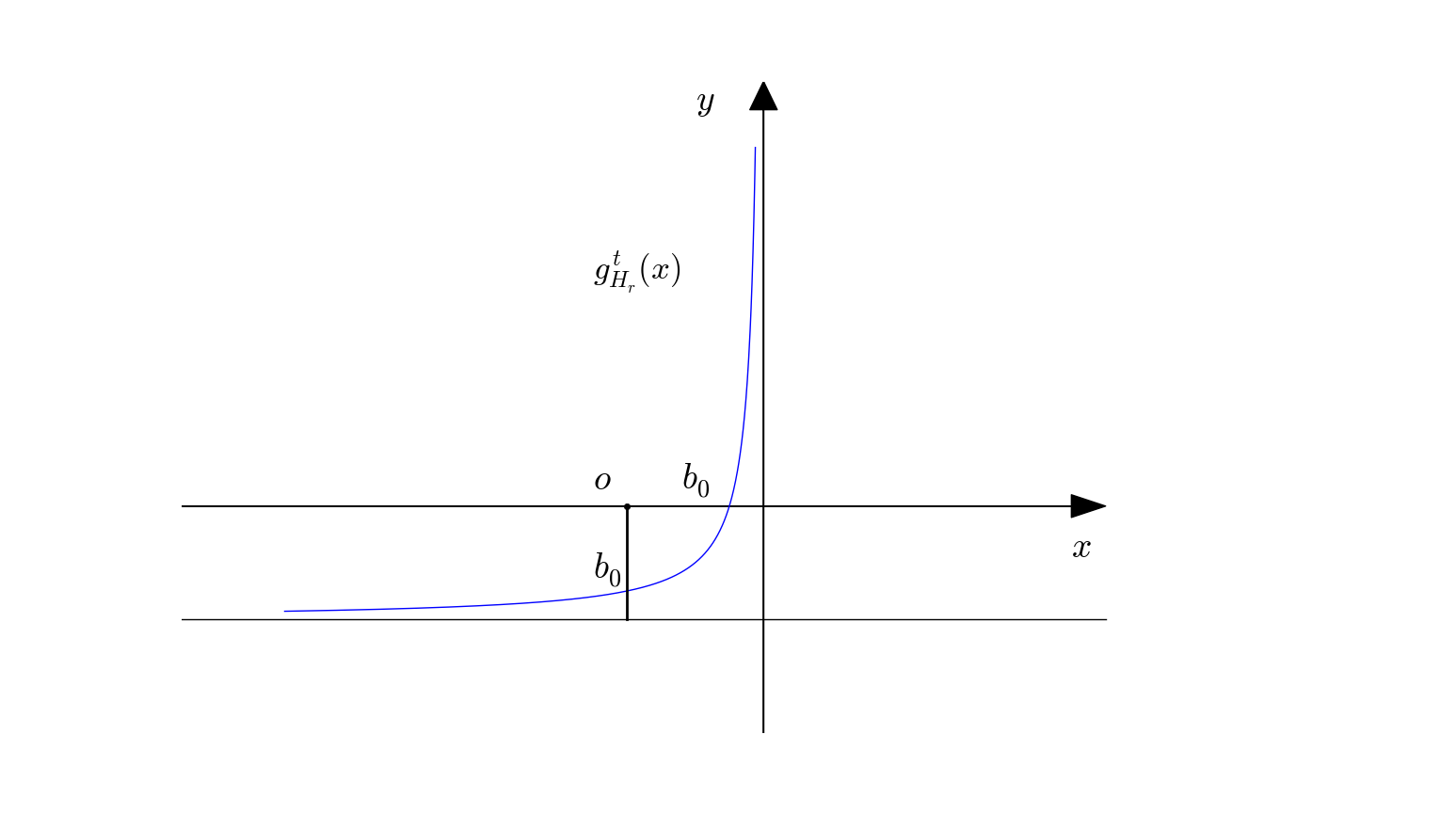}
\end{center}
\caption{Kepler trajectory $g^t_{H_r}(x)$ in the $z = z_0$ plane.  } 
\label{keplerxy}
\end{figure}

\textit{Necessity.}
Without loss of generality  $\mu_2 \le \mu_1$. Let
$$
H_r = \frac{1}{2}p^2 - \dfrac{\mu}{r},
$$
where $r \colon \mathbb R^3\setminus \{o\} \to \mathbb R$ is the distance to some point $o \in \mathbb R^3$, be a reference Hamiltonian of $F$. 
We have to show that
\begin{itemize}
 \item[{\bf 1.}] $\mu > 0$ implies $o = o_1$ and $\mu = \mu_1-\mu_2$;
 \item[{\bf 2.}] $\mu < 0$ implies $o = o_2$ and $\mu = \mu_2-\mu_1$;
 \item[{\bf 3.}] $\mu = 0$ implies $\mu_1 = \mu_2$.
\end{itemize}

{\bf Case 1.}
First we show that $o$ belongs to the $z$ axis. If this is not the case, then, due to rotational symmetry,
we have a reference Hamiltonian $H_r$ with $o = (-b_0,0,z_0)$ for some $b_0, \ z_0 \in \mathbb R, \ b_0 \ne 0$. This reference Hamiltonian $H_r$ has a trajectory 
$t \mapsto g^t_{H_{r}}(x)$ that (in the configuration space) has the form shown in
Figure~\ref{keplerxy}. But for such a trajectory
$$
L_z\left(\lim\limits_{t\to+\infty} g^t_{H_{r_1}}(x)\right) = 0 \ne \sqrt{2h} \cdot b_0 = L_z\left(\lim\limits_{t\to-\infty} g^t_{H_{r_1}}(x)\right),
$$
where $h = H_r(x) > 0$ is the energy of $g^t_{H_r}(x)$. We conclude that $o = (0,0,b)$ for some $b \in \mathbb R$.

Next we show that $b\mu = a(\mu_1-\mu_2)$. Consider a trajectory  $g^t_{H_r}(x)$ of $H_r$ that has the form shown in
Figure~\ref{keplerxz}$a$. 
\begin{figure}[ht]
\begin{center}
\includegraphics[width=1\linewidth]{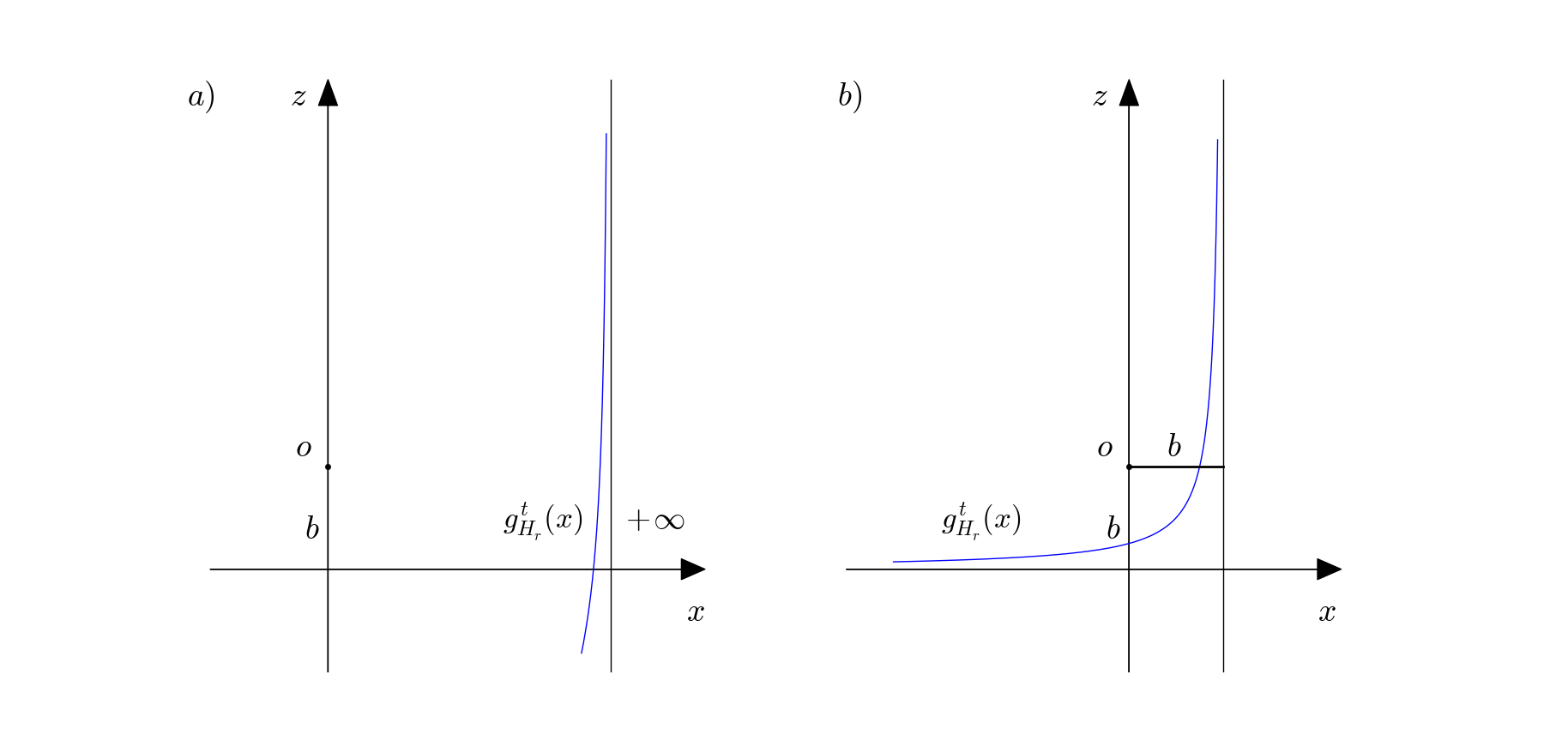}
\end{center}
\caption{Kepler trajectories in the $y = 0$ plane.  } 
\label{keplerxz}
\end{figure}
It follows from Eq.~\eqref{eq/thirdintegral} that the function
 \begin{equation*}
G_{r} = H_{r} + \frac{1}{2}(L^2-b^2(p_x^2+p_y^2)) + b(z+b)\frac{\mu}{r}
\end{equation*}
is constant along this trajectory. Thus, for $H_r$ to be a reference Hamiltonian we must have
\begin{equation} \label{eq/gmingr}
(G-G_r)\left(\lim\limits_{t\to+\infty} g^t_{H_{r_1}}(x)\right) = (G-G_r)\left(\lim\limits_{t\to-\infty} g^t_{H_{r_1}}(x)\right).
\end{equation}
In the configuration space, $g^t_{H_r}(x)$ is asymptotic
to the ray $x = c, \ y = 0, \ z \ge 0$ at  $t = +\infty$. The other asymptote at $t = -\infty$ 
gets arbitrarily close to the ray $x = c, \ y = 0, \ z \le 0$ when $c \to +\infty$. It follows that Eq.~\eqref{eq/gmingr} is equivalent to
$$
a(\mu_1 - \mu_2) - b \mu = b\mu - a(\mu_1 - \mu_2) + \varepsilon,
$$
where $\varepsilon \to 0$ when $c \to +\infty$.

The remaining equality $b = a$ can be proven using a trajectory  $g^t_{H_r}(x)$ that has the form shown in
Figure~\ref{keplerxz}$b$. 

{\bf Case 2.} In this
case  trajectories $g^t_{H_r}(x)$ 
of the repulsive Kepler Hamiltonian $H_r$ 
do not project to the
curves shown in Figs.~\ref{keplerxy}, \ref{keplerxz}$a$
and \ref{keplerxz}$b$. However, each of these curves is a branch of a hyperbola. The `complementary' branches are  (projections of) 
trajectories of $H_r$; see Fig.~\ref{keplerxy2}. If the latter branches are used, the proof becomes similar to {\bf Case 1}.
\begin{figure}[ht]
\begin{center}
\includegraphics[width=\linewidth]{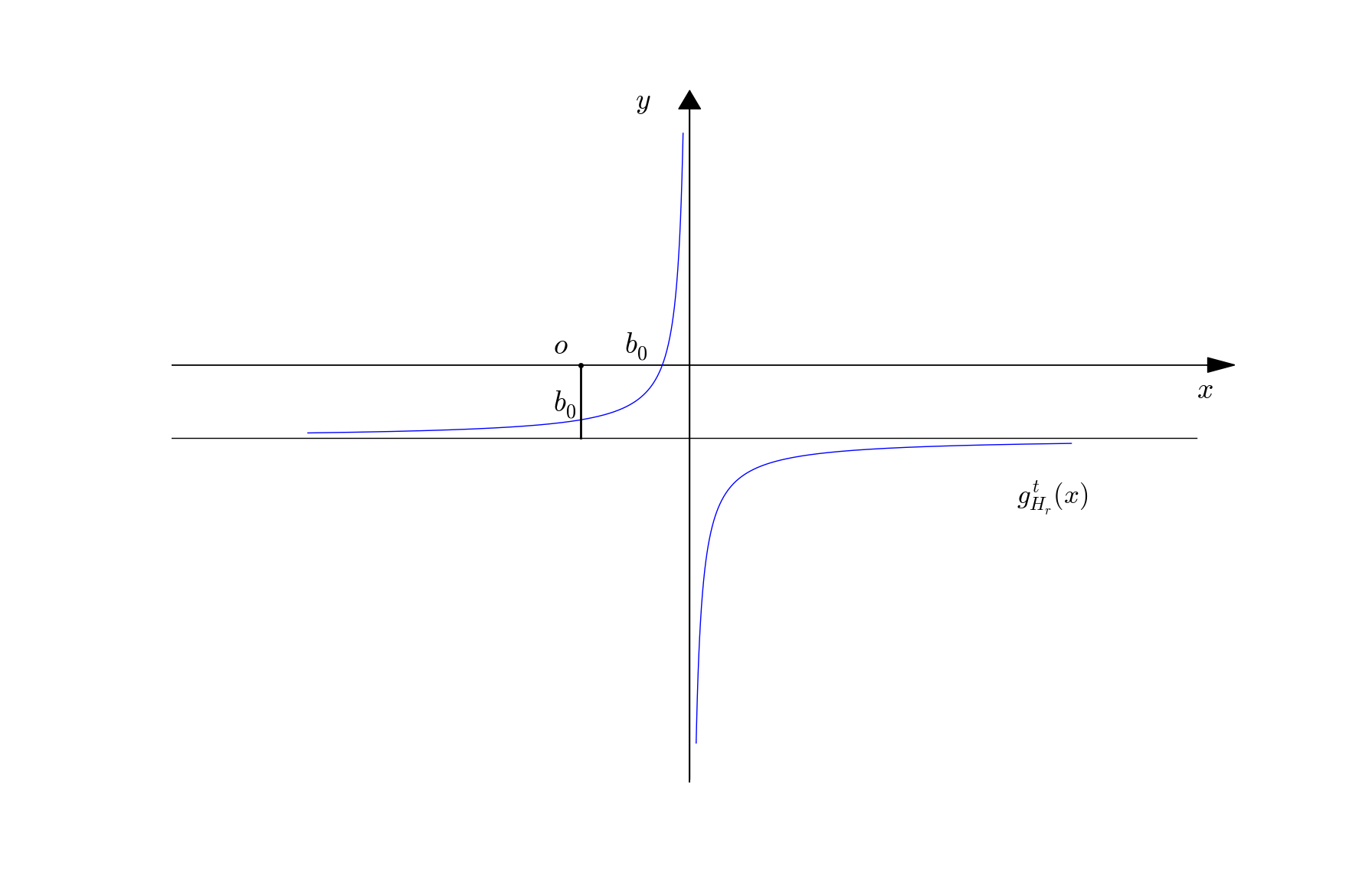}
\end{center}
\caption{The two branches ($z = z_0$ plane). In the repulsive case $\mu > 0$ a Kepler trajectory is represented by the convex branch.  } 
\label{keplerxy2}
\end{figure}

{\bf Case 3.} In this case $H_r$ generates the free motion. Let
$$g^t_{H_r}(x) = (q(t),p(t)), \  q(t) = (c,0,t), \ p(t) = (0,0,1).$$
Since $L^2$ and $(p_x,p_y,p_z)$ are conserved,
\begin{equation*}
G\left(\lim\limits_{t\to+\infty} g^t_{H_{r_1}}(x)\right) = G\left(\lim\limits_{t\to-\infty} g^t_{H_{r_1}}(x)\right)
\end{equation*}
implies $a(\mu_1-\mu_2) = a(\mu_2-\mu_1)$ and hence $\mu_1 = \mu_2$.
\end{proof}

\bibliographystyle{amsplain}
\bibliography{library}

\providecommand{\bysame}{\leavevmode\hbox to3em{\hrulefill}\thinspace}
\providecommand{\MR}{\relax\ifhmode\unskip\space\fi MR }
\providecommand{\MRhref}[2]{%
  \href{http://www.ams.org/mathscinet-getitem?mr=#1}{#2}
}
\providecommand{\href}[2]{#2}
\begin{thebibliography}{10}

\bibitem{Arnold1968}
V.I. Arnol'd and A.~Avez, \emph{Ergodic problems of classical mechanics}, W.A.
  Benjamin, Inc., 1968.

\bibitem{Bates2007}
L.~Bates and R.~Cushman, \emph{Scattering monodromy and the {A}1 singularity},
  Central European Journal of Mathematics \textbf{5} (2007), no.~3, 429--451.

\bibitem{Bolsinov2004}
A.V. Bolsinov and A.T. Fomenko, \emph{Integrable {H}amiltonian {S}ystems:
  {G}eometry, {T}opology, {C}lassification}, CRC Press, 2004.

\bibitem{BMF1990}
A.V. Bolsinov, S.V. Matveev, and Fomenko A.T., \emph{Topological classification
  of integrable {H}amiltonian systems with two degrees of freedom. list of
  systems of small complexity}, Russian Mathematical Surveys \textbf{45}
  (1990), no.~2, 59.

\bibitem{Charlier1902}
C.L. Charlier, \emph{Die {M}echanik des {H}immels}, Veit and Comp, 1902.

\bibitem{Cook1967}
J.M. Cook, \emph{{Banach algebras and asymptotic mechanics}}, {Application of
  Mathematics to Problems in Theoretical Physics: Proceedings, Summer School of
  Theoretical Physics}, vol.~6, 1967, pp.~209--246.

\bibitem{Cushman1988}
R.~Cushman and J.J. Duistermaat, \emph{The quantum mechanical spherical
  pendulum}, Bulletin of the American Mathematical Society \textbf{19} (1988),
  no.~2, 475--479.

\bibitem{Cushman2015}
R.H. Cushman and L.M. Bates, \emph{Global aspects of classical integrable
  systems}, 2 ed., Birkh{\"a}user, 2015.

\bibitem{Cushman2002}
R.H. Cushman and S.~{V{\~u} Ng{\d o}c}, \emph{Sign of the monodromy for
  {L}iouville integrable systems}, Annales Henri Poincar\'{e} \textbf{3}
  (2002), no.~5, 883--894.

\bibitem{Deprit1962}
A.~Deprit, \emph{Le probl\`{e}me des deux centres fixes}, Bull. Soc. Math. Belg
  \textbf{14} (1962), no.~11, 12--45.

\bibitem{Derezinski2013}
J.~Derezinski and C.~Gerard, \emph{Scattering theory of classical and quantum
  n-particle systems}, Theoretical and Mathematical Physics, Springer Berlin
  Heidelberg, 2013.

\bibitem{Duistermaat1980}
J.~J. Duistermaat, \emph{On global action-angle coordinates}, Communications on
  Pure and Applied Mathematics \textbf{33} (1980), no.~6, 687--706.

\bibitem{DullinWaalkens2008}
H.~Dullin and H.~Waalkens, \emph{Nonuniqueness of the phase shift in central
  scattering due to monodromy}, Phys. Rev. Lett. \textbf{101} (2008).

\bibitem{Dullin2016}
H.~R. Dullin and R.~Montgomery, \emph{Syzygies in the two center problem},
  Nonlinearity \textbf{29} (2016), no.~4, 1212.

\bibitem{DullinWaalkens2018}
Holger~R. Dullin and Holger Waalkens, \emph{Defect in the joint spectrum of
  hydrogen due to monodromy}, Phys. Rev. Lett. \textbf{120} (2018), 020507.

\bibitem{Efstathiou2017}
K.~Efstathiou, A.~Giacobbe, P.~Marde\v{s}i\'{c}, and D.~Sugny, \emph{Rotation
  forms and local {H}amiltonian monodromy}, Submitted (2016).

\bibitem{EfstathiouMartynchuk2017}
K.~Efstathiou and N.~Martynchuk, \emph{Monodromy of {H}amiltonian systems with
  complexity-1 torus actions}, Geometry and Physics \textbf{115} (2017),
  104--115.

\bibitem{Erikson1949}
H.A. Erikson and E.L. Hill, \emph{A note on the one-electron states of diatomic
  molecules}, Phys. Rev. \textbf{75} (1949), 29--31.

\bibitem{Euler1760}
L.~Euler, \emph{Probleme. {U}n corps {\' e}tant attir{\' e} en raison r{\'
  e}ciproque quarr{\' e}e des distances vers deux points fixes donn{\' e}s,
  trouver les cas o{\` u} la courbe d{\' e}crite par ce corps sera alg{\'
  e}brique}, Histoire de {L}'{A}cad{\' e}mie {R}oyale des sciences et
  {B}elles-lettres \textbf{XVI} ((1760), 1767), 228--249.

\bibitem{Euler1766}
\bysame, \emph{De motu corporis ad duo centra virium fixa attracti}, Novi
  {C}ommentarii academiae scientiarum {P}etropolitanae \textbf{10} (1766),
  207--242.

\bibitem{Euler1767}
\bysame, \emph{De motu corporis ad duo centra virium fixa attracti}, Novi
  {C}ommentarii academiae scientiarum {P}etropolitanae \textbf{11} (1767),
  152--184.

\bibitem{Flashka1988}
H.~Flaschka, \emph{A remark on integrable {H}amiltonian systems}, Physics
  Letters A \textbf{131} (1988), no.~9, 505 -- 508.

\bibitem{Fomenko1986}
A.T. Fomenko, \emph{Morse theory of integrable {H}amiltonian systems}, Dokl.
  Akad. Nauk SSSR, vol. 287, 1986, pp.~1071--1075.

\bibitem{Fomenko1987}
\bysame, \emph{The topology of surfaces of constant energy in integrable
  {H}amiltonian systems, and obstructions to integrability}, Izvestiya:
  Mathematics \textbf{29} (1987), no.~3, 629--658.

\bibitem{Fomenko1990}
A.T. Fomenko and H.~Zieschang, \emph{{Topological invariant and a criterion for
  equivalence of integrable {H}amiltonian systems with two degrees of
  freedom}}, {Izv. Akad. Nauk SSSR, Ser. Mat.} \textbf{54} (1990), no.~3,
  546--575 (Russian).

\bibitem{Gerasimov2007}
I.~A. Gerasimov, \emph{Euler problem of two fixed centers}, Friazino, (2007),
  (in Russian).

\bibitem{Herbst1974}
I.W. Herbst, \emph{{Classical scattering with long range forces}}, {Comm. Math.
  Phys.} \textbf{35} (1974), no.~3, 193--214.

\bibitem{Hunziker1968}
W.~Hunziker, \emph{The {S}-matrix in classical mechanics}, Comm. Math. Phys.
  \textbf{8} (1968), no.~4, 282--299.

\bibitem{Jacobi1884}
\~C. G.~J. Jacobi, \emph{Vorlesungen {\" u}ber {D}ynamik}, Chelsea Publ., New
  York, 1884.

\bibitem{Kim2017}
S.~Kim, \emph{Homoclinic orbits in the {E}uler problem of two fixed centers},
  https://arxiv.org/abs/1606.05622 (2017).

\bibitem{Klein2008}
M.~Klein and A.~Knauf, \emph{Classical {P}lanar {S}cattering by {C}oulombic
  {P}otentials}, Lecture Notes in Physics Monographs, Springer Berlin
  Heidelberg, 2008.

\bibitem{Knauf1999}
A.~Knauf, \emph{Qualitative aspects of classical potential scattering}, Regul.
  Chaotic Dyn. \textbf{4} (1999), no.~1, 3--22.

\bibitem{Knauf2011}
\bysame, \emph{Mathematische {P}hysik}, Springer-Lehrbuch Masterclass, Springer
  Berlin Heidelberg, 2011.

\bibitem{Knauf2008}
A.~Knauf and M.~Krapf, \emph{The non-trapping degree of scattering},
  Nonlinearity \textbf{21} (2008), no.~9, 2023.

\bibitem{Kudryavtseva2011}
E.A. Kudryavtseva and T.A. Lepskii, \emph{The topology of {L}agrangian
  foliations of integrable systems with hyperelliptic {H}amiltonian}, Sbornik:
  Mathematics \textbf{202} (2011), no.~3, 373.

\bibitem{Lagrange1766-69}
J.L. Lagrange, \emph{Miscellania taurinensia}, Recherches sur la mouvement d'un
  corps qui est attir{\' e} vers deux centres fixes \textbf{14} (1766-69).

\bibitem{Lerman1994}
L.M. Lerman and Ya.L. Umanski{\u \i}, \emph{Classification of four-dimensional
  integrable {H}amiltonian systems and {P}oisson actions of {$\mathbb{R}^2$} in
  extended neighborhoods of simple singular points. i}, Russian Academy of
  Sciences. Sbornik Mathematics \textbf{77} (1994), no.~2, 511--542.

\bibitem{Martynchuk2017}
N.~Martynchuk and K.~Efstathiou, \emph{Parallel transport along {S}eifert
  manifolds and fractional monodromy}, Communications in Mathematical Physics
  \textbf{356} (2017), no.~2, 427--449.

\bibitem{Martynchuk2016}
N.~Martynchuk and H.~Waalkens, \emph{Knauf's degree and monodromy in planar
  potential scattering}, Regular and Chaotic Dynamics \textbf{21} (2016),
  no.~6, 697--706.

\bibitem{Martynchuk2015}
N.N. Martynchuk, \emph{Semi-local {L}iouville equivalence of complex
  {H}amiltonian systems defined by rational {H}amiltonian}, Topology and its
  Applications \textbf{191} (2015), no.~Supplement C, 119 -- 130.

\bibitem{Matveev1996}
V.S. Matveev, \emph{Integrable {H}amiltonian system with two degrees of
  freedom. {T}he topological structure of saturated neighbourhoods of points of
  focus-focus and saddle-saddle type}, Sbornik: Mathematics \textbf{187}
  (1996), no.~4, 495--524.

\bibitem{Nekhoroshev2006}
N.N. Nekhoroshev, D.A. Sadovski\'{i}, and B.I. Zhilinski\'{i}, \emph{Fractional
  {H}amiltonian monodromy}, Annales Henri Poincar\'{e} \textbf{7} (2006),
  1099--1211.

\bibitem{Niessen1923}
K.F. Niessen, \emph{Zur {Q}uantentheorie des {W}asserstoffmolek{\"u}lions},
  Annalen der Physik \textbf{375} (1923), no.~2, 129--134.

\bibitem{OMathuna2008}
D.~{\'O}'Math{\' u}na, \emph{Integrable systems in celestial mechanics},
  Birkh{\" a}user, Basel, 2008.

\bibitem{Pauli1922}
W.~Pauli, \emph{{\"U}ber das {M}odell des {W}asserstoffmolek{\"u}lions},
  Annalen der Physik \textbf{373} (1922), no.~11, 177--240.

\bibitem{Seri2015}
M.~Seri, \emph{The problem of two fixed centers: bifurcation diagram for
  positive energies}, Journal of Mathematical Physics \textbf{56} (2015),
  no.~1, 012902.

\bibitem{Seri2016}
M.~Seri, A.~Knauf, M.~D. Esposti, and T.~Jecko, \emph{Resonances in the
  two-center {C}oulomb systems}, Reviews in Mathematical Physics \textbf{28}
  (2016), no.~07, 1650016.

\bibitem{Simon1971}
B.~Simon, \emph{{Wave operators for classical particle scattering}}, {Comm.
  Math. Phys.} \textbf{23} (1971), no.~1, 37--48.

\bibitem{Vu-Ngoc1999}
S.~{V{\~u} Ng{\d o}c}, \emph{Quantum monodromy in integrable systems},
  Communications in Mathematical Physics \textbf{203} (1999), no.~2, 465--479.

\bibitem{Vosmischera2003}
T.G. Vosmischera, \emph{Integrable systems of celestial mechanics in space of
  constant curvature}, Springer {N}etherlands, 2003.

\bibitem{WaalkensDullinRichter04}
H.~Waalkens, H.R. Dullin, and P.H. Richter, \emph{The problem of two fixed
  centers: bifurcations, actions, monodromy}, Physica D \textbf{196} (2004),
  no.~3-4, 265--310.

\bibitem{Whittaker1917}
E.T. Whittaker, \emph{A treatise on the analytical dynamics of particles and
  rigid bodies; with an introduction to the problem of three bodies},
  Cambridge, University Press, 1917.

\bibitem{Zagryadskii2012}
O.A. Zagryadskii, E.A. Kudryavtseva, and D.A. Fedoseev, \emph{A generalization
  of {B}ertrand's theorem to surfaces of revolution}, Sbornik: Mathematics
  \textbf{203} (2012), no.~8, 1112.

\bibitem{Zung1997}
N.T. Zung, \emph{A note on focus-focus singularities}, Differential Geometry
  and its Applications \textbf{7} (1997), no.~2, 123--130.

\end{thebibliography}

\end{document}